\documentstyle[10pt,amscd,xypic,amssymb,combelow,wasysym]{amsart}

\xyoption{all}
\CompileMatrices

\newcommand{\nc}{\newcommand}
\newcommand{\rnc}{\renewcommand}

\makeatletter
\@addtoreset{equation}{section}
\makeatother

\newenvironment{proof}{{\noindent \textbf{Proof}\,\,}}{\hspace*{\fill}$\Box$\medskip}

\newtheorem{theorem}[equation]{Theorem}
\newtheorem{proposition}[equation]{Proposition}
\newtheorem{lemma}[equation]{Lemma}
\newtheorem{corollary}[equation]{Corollary}

\newtheorem{claim}[equation]{Claim}

\theoremstyle{definition}

\newtheorem{definition}[equation]{Definition}

\theoremstyle{remark}

\newtheorem{remark}[equation]{Remark}

\nc{\fa}{{\mathfrak{a}}}
\nc{\fb}{{\mathfrak{b}}}
\nc{\fg}{{\mathfrak{g}}}
\nc{\fh}{{\mathfrak{h}}}
\nc{\fj}{{\mathfrak{j}}}
\nc{\fn}{{\mathfrak{n}}}
\nc{\fu}{{\mathfrak{u}}}
\nc{\fp}{{\mathfrak{p}}}
\nc{\fr}{{\mathfrak{r}}}
\nc{\ft}{{\mathfrak{t}}}
\nc{\fsl}{{\mathfrak{sl}}}
\nc{\fgl}{{\mathfrak{gl}}}
\nc{\hsl}{{\widehat{\mathfrak{sl}}}}
\nc{\hgl}{{\widehat{\mathfrak{gl}}}}

\nc{\fB}{{\mathfrak{B}}}
\nc{\fC}{{\mathfrak{C}}}
\nc{\fZ}{{\mathfrak{Z}}}
\nc{\fW}{{\mathfrak{W}}}

\nc{\pol}{{\text{Poles}}}

\nc{\BA}{{\mathbb{A}}}
\nc{\BC}{{\mathbb{C}}}
\nc{\BK}{{\mathbb{K}}}
\nc{\BM}{{\mathbb{M}}}
\nc{\BN}{{\mathbb{N}}}
\nc{\BF}{{\mathbb{F}}}
\nc{\BQ}{{\mathbb{Q}}}
\nc{\BP}{{\mathbb{P}}}
\nc{\BR}{{\mathbb{R}}}
\nc{\BZ}{{\mathbb{Z}}}

\nc{\B}{{\mathcal{B}}}
\nc{\E}{{\mathcal{E}}}
\nc{\F}{{\mathcal{F}}}
\nc{\K}{{\mathcal{K}}}
\rnc{\L}{{\mathcal{L}}}
\nc{\M}{{\mathcal{M}}}
\nc{\N}{{\mathcal{N}}}
\rnc{\O}{{\mathcal{O}}}
\rnc{\S}{{\mathcal{S}}}
\nc{\T}{{\mathcal{T}}}
\nc{\V}{{\mathcal{V}}}
\nc{\Y}{{\mathcal{Y}}}

\nc{\tS}{{\widetilde{\S}}}

\nc{\uu}{{U_q(\hgl_n)}}
\nc{\uuo}{{U^0_q(\hgl_n)}}
\nc{\uug}{{U_q^+(\hgl_n)}}
\nc{\uul}{{U_q^-(\hgl_n)}}

\nc{\su}{{U_q(\hsl_n)}}
\nc{\suo}{{U^0_q(\hsl_n)}}
\nc{\sug}{{U_q^+(\hsl_n)}}
\nc{\sul}{{U_q^-(\hsl_n)}}

\nc{\uui}{{U_q(\hgl_1)}}
\nc{\uuio}{{U^0_q(\hgl_1)}}
\nc{\uuig}{{U_q^+(\hgl_1)}}
\nc{\uuil}{{U_q^-(\hgl_1)}}

\nc{\yy}{{Y_{t_1,t_2}(\hgl_1)}}
\nc{\yyo}{{Y^0_{t_1,t_2}(\hgl_1)}}
\nc{\yyp}{{Y^+_{t_1,t_2}(\hgl_1)}}
\nc{\yym}{{Y^-_{t_1,t_2}(\hgl_1)}}
\nc{\yyg}{{Y^\geq_{t_1,t_2}(\hgl_1)}}
\nc{\yyl}{{Y^\leq_{t_1,t_2}(\hgl_1)}}

\nc{\e}{{\varepsilon}}
\nc{\nn}{{\mathbb{N}}^n}
\nc{\Tr}{{\text{Tr}}}
\nc{\tT}{{T}}

\nc{\od}{{\overline{d}}}
\nc{\rg}{{\textrm{R}\Gamma}}
\nc{\erg}{{\emph{R}\Gamma}}
\nc{\id}{{\textrm{id}}}

\def\bc{{\mathbf{c}}}
\def\bd{{\mathbf{d}}}
\def\be{{\mathbf{e}}}

\def\bu{{\mathbf{u}}}

\def\Ext{\textrm{Ext}}
\def\Hom{\textrm{Hom}}

\def\la{{\lambda}}
\def\lamu{{\lambda \backslash \mu}}
\def\lanu{{\lambda \backslash \nu}}
\def\numu{{\nu \backslash \mu}}

\def\sq{{\square}}
\def\bsq{{\blacksquare}}

\def\bla{{\boldsymbol{\la}}}

\def\bmu{{\boldsymbol{\mu}}}
\def\bnu{{\boldsymbol{\nu}}}
\def\blamu{{\boldsymbol{\lamu}}}
\def\blanu{{\boldsymbol{\lanu}}}

\def\bnumu{{\boldsymbol{\numu}}}

\def\tab{{\text{} \\}}

\def\of{\overline{f}}

\def\Ext{{\text{Ext}}}
\def\id{{\text{Id}}}

\def\sym{{\text{Sym}}}
\def\esym{{\emph{Sym}}}
\def\res{{\text{Res}}}

\def\esyt{{\emph{SYT}}}

\def\esh{{\emph{sh}}}
\def\sh{{\text{sh}}}
\def\sma{{\text{small}}}
\def\ideg{{\text{--deg}}}
\def\edeg{{\emph{--deg}}}

\nc{\loccitt}{{\emph{loc. cit.}}}
\nc{\loccit}{{\emph{loc. cit. }}}
\nc{\wheel}{{\text{wheel}}}
\nc{\totdeg}{{\text{totdeg}}}
\nc{\loc}{{\text{loc}}}
\nc{\im}{{\text{Im }}}
\nc{\eim}{{\emph{Im }}}

\nc{\vir}{{\text{vir}}}
\nc{\evir}{{\emph{vir}}}

\nc{\sgn}{{\text{sign}}}

\nc{\baru}{{\bar{u}}}
\nc{\barbu}{{\bar{\bu}}}

\nc{\walpha}{{\widetilde{\alpha}}}
\nc{\wbeta}{{\widetilde{\beta}}}
\nc{\wgamma}{{\widetilde{\gamma}}}

\nc{\wS}{{\widetilde{\S}}}
\nc{\vac}{{|\emptyset\rangle}}
\nc{\wOmega}{{\widetilde{\Omega}}}

\rnc{\be}{{\bar{\e}}}
\nc{\lev}{{\text{level}}}
\nc{\wi}{{\text{without }\infty}}

\begin{document}

\title[Exts and the AGT relations]{\Large{\textbf{Exts and the AGT relations}}}

\author[Andrei Negu\cb t]{Andrei Negu\cb t}
\address{Massachusetts Institute of Technology, Mathematics Department, Cambridge, USA}
\address{Simion Stoilow Institute of Mathematics, Bucure{\cb s}ti, Romania}
\email{andrei.negut@@gmail.com}

\begin{abstract}

We prove the connection between the Nekrasov partition function of $\N=2$ super--symmetric $U(2)$ gauge theory with adjoint matter and conformal blocks for the Virasoro algebra, as predicted by the Alday--Gaiotto--Tachikawa relations. Mathematically, this is achieved by relating the Carlsson--Okounkov Ext vector bundle on the moduli space of rank 2 sheaves with Liouville vertex operators. Our approach is geometric in nature, and uses a new method for intersection--theoretic computations of the Ext operator.

\end{abstract}

\keywords{Nekrasov partition function, AGT relations, Ext operators, shuffle algebra} 
\subjclass{14D21, 81T13, 81R10}

\maketitle

\section{Introduction}

\noindent
Fix a natural number $r$. The $\Ext$ bundle was defined in \cite{CO} as:
\begin{equation}
\label{eqn:ext}
\xymatrix{ & \E \ar@{-->}[d] \\
& \M_{r,d} \times \M_{r,d'} \ar[ld]_{p^1} \ar[rd]^{p^2} \\ 
\M_{r,d} & & \M_{r,d'}} \qquad \xymatrix{ \\ \E|_{\F,\F'} = \Ext^1(\F',\F(-\infty))}
\end{equation}
where $\M_{r,d}$ denotes the moduli space of rank $r$ degree $d$ torsion free sheaves on $\BP^2$, framed at a fixed line $\infty \subset \BP^2$. The Chern polynomial of $\E$ induces an operator:
\begin{equation}
\label{eqn:extoperator0}
A_m|_{d}^{d'} \ : \ H_{\bu',d'} \longrightarrow H_{\bu,d}
\end{equation}
$$
\alpha \mapsto p^1_* \Big(c(\E,m)\cdot p^{2*}(\alpha) \Big)
$$
where the equivariant cohomology groups are defined as:
$$
H_{\bu,d} = H_{\BC^* \times \BC^* \times (\BC^*)^r}(\M_{r,d})
$$
We write $t_1,t_2,u_1,...,u_r$ for the equivariant parameters of the torus $\BC^* \times \BC^* \times (\BC^*)^r$, and encode the latter $r$ of these in the vector $\bu = (u_1,...,u_r)$. As in the work of Nakajima and Grojnowski, it makes sense to group all the cohomologies together:
$$
H_\bu = \bigoplus_{d=0}^\infty H_{\bu,d}
$$
With this in mind, \loccit define the {\bf Ext operator} as:
\begin{equation}
\label{eqn:extoperator}
A_m(x) \ : \ H_{\bu'} \longrightarrow H_{\bu}
\end{equation}
$$
A_m(x) = \sum_{d,d' = 0}^\infty  A_m|_d^{d'} \cdot x^{d - d' + \lev_\bu - \lev_{\bu'}}
$$
where the level of the representation $H_\bu$ is a constant such that the 0--th Virasoro mode acts by $\bd + \lev_\bu$. The degree operator $\bd$ acts with eigenvalue $d$ on the graded piece $H_{\bu,d}$. See \eqref{eqn:cartan2} for the explicit formula of the level in our notation. Elements of the space $H_\bu$ can be thought of as vectors of rational functions of $t_1,t_2,\bu$. The matrix coefficients of the operator $A_m(x)$ are rational functions of $t_1,t_2,\bu-\bu',m,x$, since we think of two different rank $r$ tori acting on the moduli spaces $\M_{r,d}$ and $\M_{r,d'}$ in \eqref{eqn:ext}. In particular, we will study the generating series:
\begin{equation}
\label{eqn:nekrasov}
Z_{m_1,...,m_k}(x_1,...,x_k) = \Tr \left(Q^\bd A_{m_1}(x_1) ... A_{m_k}(x_k) \Big |_{H_{\bu^{k+1}} \rightarrow H_{\bu^1}} \right) 
\end{equation}
For the right hand side to make sense as a trace, we should visualize $A_{m_i}(x_i)$ as a morphism $H_{\bu^{i+1}} \rightarrow H_{\bu^{i}}$ for collections of parameters $\{\bu^i\}_{1\leq i \leq k+1}$, and identify $\bu^{k+1} = \bu^1$. Therefore, \eqref{eqn:nekrasov} is a function of $Q,t_1,t_2,\bu^1,...\bu^{k},m_1,...,m_k,x_1,...,x_k$.

\tab 
The generating function \eqref{eqn:nekrasov} is the Nekrasov partition function of $U(r)$ gauge theory with adjoint matter. It was introduced by Nekrasov in \cite{nek} and developed further by Nekrasov and Okounkov in \cite{no}. The philosophy behind \eqref{eqn:nekrasov}, namely the fact that one can realize the partition function $Z$ as the trace of Ext operators, goes back to work of Carlsson, Okounkov and Nekrasov. In the present paper, we are interested in the Alday--Gaiotto--Tachikawa relations of \cite{AGT}, presented in \eqref{eqn:agt} below, which predict that $Z$ is related to a Liouville conformal block arising as a correlation function in a certain conformal field theory on a sphere with $k$ punctures.

\tab 
Let us explain this connection in mathematical language. We answer a question posed by Carlsson in \cite{C} that frames the AGT relations \eqref{eqn:agt} for $Z$ as an equality of operators. More precisely, the AGT relations follow once we connect the Ext operator $A_m(x)$ to certain ``intertwiners" $\Omega_m(x)$ for the algebra $W_r = W(\fgl_r)$:
\begin{equation}
\label{eqn:connection}
\xymatrixcolsep{6pc}\xymatrix{
A_m(x) \ \ \ar@{<~>}[r]^{\text{roughly equal to}} & \ \ \Omega_m(x)} 
\end{equation}
Here, $\Omega_m(x)$ denotes a certain operator between universal Verma modules of $W_r$:
\begin{equation}
\label{eqn:conformalblocks}
\Omega_m(x) : M_{\bu'} \longrightarrow M_\bu
\end{equation}
which has a prescribed interaction with the generators of the $W$--algebra. In particular, $\Omega_m(x)$ must satisfy the following commutation relations with the Heisenberg--Virasoro subalgebra $\{B_k,L_k\}_{k\in \BZ} \subset W_r$ (see Theorem \ref{thm:full} for our conventions):
\begin{equation}
\label{eqn:inter1}
\Big[ B_{k}, \Omega_m(x) \Big] = \beta x^k \cdot \Omega_m(x) 
\end{equation}
\begin{equation}
\label{eqn:inter2}
\Big[ L_{k}, \Omega_m(x) \Big] = \left(x^{k+1} \frac {\partial}{\partial x} - \lambda k x^k \right)\cdot \Omega_m(x) 
\end{equation}
for all $k \in \BZ$, where we define the constants $t=t_1+t_2$ and:
\begin{equation}
\label{eqn:beta}
\beta = |\bu'| - |\bu| = \sum_{i=1}^r \left( u'_i - u_i \right)
\end{equation}
\begin{equation}
\label{eqn:lambda}
\lambda = \frac {r(r-1)m(m-t) - (r-1)(2m-t) \beta + \beta^2}{2t_1t_2}
\end{equation}
When $r = 2$, relations \eqref{eqn:inter1}--\eqref{eqn:inter2} determine the operator $\Omega_m(x)$ uniquely, and it is known as a {\bf Liouville vertex operator}. The AGT relations claim that the Nekrasov partition function \eqref{eqn:nekrasov} coincides with $\Tr\left(Q^\bd\Omega_{m_1}(x_1)...\Omega_{m_k}(x_k) \right)$, up to a prefactor which we will make explicit in Corollary \ref{cor:agt}. The strategy to prove this equality is based on the following observation. There exists a geometric action $W_r \curvearrowright H_\bu$, defined independently and by different means in \cite{MO} and \cite{SV1}, such that:
\begin{equation}
\label{eqn:isomorphism}
H_\bu \cong M_\bu
\end{equation}
Our main result is Theorem \ref{thm:main}, which claims that the Ext operator enjoys similar properties with \eqref{eqn:inter1}--\eqref{eqn:inter2}. When $r=2$, this will allow us to establish the connection \eqref{eqn:connection} in Corollary \ref{cor:main}. We formulate the theorem in terms of $A_m = A_m(1)$, since the variable $x$ is redundant when dealing with a single operator: \\

\begin{theorem}
\label{thm:main}
For any $r \geq 1$, we have the following commutation relations between the Ext operator $A_m : H_{\bu'} \rightarrow H_{\bu}$ and the Heisenberg--Virasoro subalgebra $\{B_k,L_k\}_{k\in \BZ} \subset W_r \curvearrowright H_\bu, H_{\bu'}$: 
\begin{equation}
\label{eqn:Inter1}
\Big[ B_{\pm k}, A_m \Big] = (\beta - r(m - t\be) ) A_m
\end{equation}
\begin{equation}
\label{eqn:Inter2}
\Big [L_{\pm k} - L_{\pm  (k-1)} , A_m \Big] =\left(\frac {m - t\be}{t_1t_2}B_{\pm (k-1)} -  \frac {m-t\e}{t_1t_2} B_{\pm k} \mp \right.
\end{equation}
$$
\left. \mp \frac {r(m^2r - (r + 1)mt + t^2\be \pm \delta_k^1(m-t\be)^2) - (2mr - (r + 1)t)\beta+ \beta^2}{2t_1t_2} \right) A_m
$$
for all $k>0$, where we write $t = t_1+t_2$, $\e = \frac {1 \pm 1}2$ and $\be = \frac {1 \mp 1}2$. 
\end{theorem}

\tab 
The connection between $A_m$ and $\Omega_m$ can be seen by introducing the series:
$$
g_+(x) = \exp\left( \sum_{k=1}^{\infty} \frac {B_kx^{-k}}k \right) \qquad \qquad g_-(x) = \exp\left( \sum_{k=1}^{\infty} \frac {B_{-k}x^k}k \right)
$$
which can be thought of as lying in a completion of the Heisenberg Lie algebra. \\

\begin{corollary}
\label{cor:main}
For any $r$, under the isomorphism $H_\bu \cong M_\bu$ of \eqref{eqn:isomorphism}, one has:
\begin{equation}
\label{eqn:factor}
A_m(x) = g_-^{\frac m{t_1t_2}}(x) \cdot \Omega_m(x) \cdot g_+^{\frac {t-m}{t_1t_2}}(x)
\end{equation}
where $\Omega_m(x):M_{\bu'} \rightarrow M_{\bu}$ is an operator that satisfies relations \eqref{eqn:inter1} and \eqref{eqn:inter2}. 

\end{corollary}

\tab 
Since properties \eqref{eqn:inter1}--\eqref{eqn:inter2} uniquely determine the operator $\Omega_m(x)$ when $r\leq 2$, we can easily identify it. When $r=1$, $\Omega_m(x)$ is a trivial shift operator which does not depend on $m$. When $r = 2$, $\Omega_m(x)$ is called the Liouville vertex operator, as explained above. When $r>2$, relations \eqref{eqn:Inter1}--\eqref{eqn:Inter2} (respectively \eqref{eqn:inter1}--\eqref{eqn:inter2}) are not enough to determine the operator $A_m$ (respectively $\Omega_m$) completely. In this case, one needs to compute commutation relations of $A_m$ with the higher currents of the $W$--algebra, and this falls outside our range of possibilities at the moment. The reason for this is that our computations use the shuffle algebra interpretation of the Schiffmann--Vasserot construction (\cite{SV1}), which allows us to write explicit formulas for geometric correspondences. For general $r$, understanding the connection between the shuffle algebra and $W_r$ is an ongoing endeavor. \\

\begin{corollary}
\label{cor:agt}
For $r=2$, we have the following formula for the function \eqref{eqn:nekrasov}:
\begin{equation}
\label{eqn:agt}
Z = \frac {\emph{Tr} \left(Q^\bd \Omega_{m_1}(x_1) ... \Omega_{m_k}(x_k) \Big |_{M_{\bu^{k+1}} \rightarrow M_{\bu^1}} \right)}{\prod_{1\leq i<j \leq k} \left( \frac {x_j}{x_i};Q \right)_\infty^{e} \prod_{1\leq j \leq i \leq k} \left( \frac {x_jQ}{x_i};Q \right)_\infty^{e}}
\end{equation}
where $(a;Q)_\infty = \prod_{i=0}^\infty (1-aQ^i)$ denotes the infinite $Q$--Pochhammer symbol and:\\
$$
e = \frac {2(m_i-t)m_j + \beta(m_i-t) + \beta m_j}{t_1t_2}
$$
\end{corollary}

\tab 
Corollaries \ref{cor:main} and \ref{cor:agt} were proved in \cite{C} for $r=2$, $t = |\bu| = |\bu'| = 0$ by different means, namely by using the Segal-Sugawara construction. Our method is intersection-theoretic, and involves setting up certain integral formulas for the composed correspondences $B_{\pm k} \circ A_m$ and $A_m \circ B_{\pm k}$. Changing contours and the residue theorem allow one to compute the difference between the two compositions, hence obtaining formula \eqref{eqn:Inter1}. Formula \eqref{eqn:Inter2} is proved analogously. We believe that our method applies to more general quiver varieties, and in particular has been used in \cite{affine} for affine Laumon spaces. There, the Nekrasov partition function of $U(r)$ gauge theory with adjoint matter and a full surface operator insertion was connected to the eigenfunction of the Calogero--Moser integrable system.

\tab 
After the publication of the present paper, Yutaka Matsuo pointed out his earlier paper \cite{KMZ} with Shoichi Kanno and Hong Zhang. In \loccitt, the authors use the Schiffmann--Vasserot algebra to obtain a system of recursion relations that completely determine the Nekrasov partition function, and conclude a result similar to Corollary \ref{cor:agt}. Like ours, their approach is to compute commutation relations of vertex operators with the Schiffmann--Vasserot algebra. However, the computational tools we employ are quite different: while \loccit uses formulas in the basis of fixed points indexed by partitions, we use the shuffle algebra and tautological classes on moduli of sheaves in order to carry out intersection--theoretic computations. To the author's knowledge, this constitutes a new way of computing the Ext operator.

\tab 
The structure of this paper is the following: in Section \ref{sec:shuffle}, we describe the shuffle algebra incarnation of the affine Yangian, with the purpose of isolating the Heisenberg--Virasoro algebra within. In Section \ref{sec:mod}, we study how the shuffle algebra acts on the cohomology rings $H_\bu$ and prove Theorem \ref{thm:main}. In Section \ref{sec:reptheory}, we use representation-theoretic techniques to prove Corollary \ref{cor:main} and Corollary \ref{cor:agt}. Many of our propositions are computations, and although they shed a lot of light on the nature of shuffle algebra calculus, we leave them for the Appendix so as to not obscure the general direction of the paper.  

\tab 
Special thanks are due to Erik Carlsson for explaining his work \cite{C}, which provides much of the inspiration for this paper. I also want to thank Pavel Etingof, Davesh Maulik, Michael McBreen, Hiraku Nakajima, Andrei Okounkov, Francesco Sala and Alexander Tsymbaliuk for many useful talks on AGT and Yangians. \\

\section{The shuffle algebra and its Heisenberg-Virasoro subalgebra} 
\label{sec:shuffle}

\subsection{}\label{sub:shuffle}

Let $\BF = \BQ(t_1,t_2)$ and consider the following rational function:
\begin{equation}
\label{eqn:zeta}
\zeta(z) = \frac {(z + t_1)(z + t_2)}{z(z + t)} \qquad \qquad \text{where} \qquad t = t_1+ t_2
\end{equation}
Note the identity:
\begin{equation}
\label{eqn:identity}
\zeta(z) = \zeta(- t - z)
\end{equation} 
and the fact that:
\begin{equation}
\label{eqn:limit} 
\zeta(z) = 1 + \frac {t_1t_2}{z(z+t)} = 1 + O \left( \frac 1{z^2} \right)
\end{equation}
Consider the vector space of symmetric rational functions:
$$
V = \bigoplus_{k=0}^\infty \BF(z_1,...,z_k)^\sym
$$
endowed with the following {\bf shuffle product}:
$$
R(z_1,...,z_k) * R'(z_1,...,z_{k'}) = \frac 1{k! \cdot k'!} \cdot 
$$
\begin{equation}
\label{eqn:shufprod}
\sym \left[ R(z_1,...,z_k) R'(z_{k+1},...,z_{k+k'}) \prod^{1 \leq i \leq k}_{k < j \leq k+k'} \zeta(z_i-z_j) \right]
\end{equation}
where $\sym$ denotes summing over all permutations of the variables $z_1,...,z_{k+k'}$. \\

\begin{definition}
\label{def:shuf}

The {\bf shuffle algebra} is defined as the following subspace of $V$:
\begin{equation}
\label{eqn:defshuf}
\S =  \bigoplus_{k=0}^\infty \frac {\Big \{ \rho(z_1,...,z_k)\text{ symmetric satisfying wheel conditions}\Big\}}{\prod_{1\leq i \neq j \leq k} (z_i - z_j + t)}
\end{equation}
where a symmetric polynomial $\rho$ is said to satisfy the {\bf wheel conditions} if:
\begin{equation}
\label{eqn:wheel}
\rho \Big |_{z_1 \mapsto y, z_2 \mapsto y + t_1, z_3 \mapsto y + t} = \rho \Big |_{z_1 \mapsto y, z_2 \mapsto y + t_2, z_3 \mapsto y + t}  = 0
\end{equation}
We write $\S_k \subset \S$ for the $k-$th direct summand of \eqref{eqn:defshuf}, and note that this grading is respected by the multiplication \eqref{eqn:shufprod}. \\
\end{definition}

\subsection{}\label{sub:minimal}

Let us consider the so-called {\bf small shuffle algebra}, defined as:
$$
\S_\sma = \Big \langle \text{subalgebra generated by } z_1^d \Big \rangle_{d\geq 0} \ \subset \ \S
$$
One of the main results of this Section is the following: \\

\begin{theorem}
\label{thm:generation}

We have $\S_{\emph{small}} = \S$.  \\

\end{theorem}

\noindent
The Theorem follows immediately from Propositions \ref{prop:belong} and \ref{prop:basis} below. A key role in the proof of Theorem \ref{thm:generation} is played by the following rational functions:
\begin{equation}
\label{eqn:ideal}
C_m = \sym \left [\frac {m(z_1,...,z_k)}{\prod_{i=1}^{k-1} (z_{i+1} - z_{i}+t)} \prod_{1 \leq i < j \leq k} \zeta(z_i-z_j) \right]
\end{equation}
as $m$ goes over all polynomials with coefficients in $\BF$, not necessarily symmetric. \\

\begin{proposition}
\label{prop:belong}
For any $m\in \BF[z_1,...,z_k]$, we have $C_m \in \S_{\emph{small}}$. \\
\end{proposition}

\subsection{}\label{sub:bdeg}
 
For any symmetric rational function $R(z_1,...,z_k)$, we define its $l$--{\bf degree} as:
\begin{equation}
\label{eqn:bdeg}
l\ideg_R := \deg_y \left( R \Big |_{z_1 \mapsto y - t_1, z_2 \mapsto y - 2 t_1,..., z_l \mapsto y - l  t_1} \right)
\end{equation}
where the degree of a rational function in a single variable $y$ is the degree of its numerator minus the degree of its denominator. The $0-$degree is always equal to $0$. Formula \eqref{eqn:limit} claims the rational function $\zeta(y)$ has degree $0$ in $y$, and so we have:
\begin{equation}
\label{eqn:degree}
l\ideg_{R * R'} \leq \max \Big\{a\ideg_R + a'\ideg_{R'} \Big\}^{a+a' = l}_{a\leq k, a' \leq k'}
\end{equation}
for any shuffle elements $R\in \S_k$ and $R'\in \S_{k'}$.  However, we will need a slight improvement of the bound \eqref{eqn:degree} in the case of a commutator: \\

\begin{lemma}
\label{lem:improve}

For any pair of shuffle elements $R\in \S_k$ and $R' \in \S_{k'}$, we have:
\begin{equation}
\label{eqn:degreeimprove}
l\edeg_{[R , R']} \leq \max \emph{ among } \Big\{a\edeg_R + a'\edeg_{R'} \Big\}^{a+a' = l}_{a < k, a' < k'}
\end{equation}
$$
\emph{and } \ \Big\{ l \edeg_R - 2 \Big\}^{\emph{if}}_{l=k} \ \ \emph{ and } \ \ \Big\{ l\edeg_{R'} - 2 \Big\}^{\emph{if}}_{l = k'} \ \emph{ and}
$$
$$
\Big\{ k\edeg_R + (l-k)\edeg_{R'} - 1 \Big\}^{\emph{if}}_{l>k} \ \emph{ and } \ \Big\{ (l-k')\edeg_R + k'\edeg_{R'} - 1 \Big\}^{\emph{if}}_{l>k'}
$$
The terms in the last two lines appear if $l=k$, $l=k'$, $l>k$, $l>k'$, respectively. \\

\end{lemma}

\subsection{}\label{sub:slope}

For any vector of real numbers $\bd = (d_1,...,d_k)$, we consider the following: \\

\begin{definition}\label{def:slope}

We say that a shuffle element $R(z_1,...,z_k)$ has {\bf slope} $\leq \bd$ if:
\begin{equation}
\label{eqn:ideg1}
l\ideg_R \leq d_l \qquad \ \qquad \forall \ l\in \{1,...,k\}
\end{equation}
We say that it has {\bf slope} $<\bd$ if:
\begin{equation}
\label{eqn:ideg2}
\begin{cases} l\ideg_R < d_l & \quad \forall \ l \in \{1,...,k-1\} \\ 
k\ideg_R \leq d_k & \end{cases}
\end{equation}
and we say that it has {\bf slope} $ \ll \bd$ if:
\begin{equation}
\label{eqn:ideg3}
\begin{cases} l\ideg_R < d_l - 1 & \forall \ l \in \{1,...,k-1\} \\ 
k\ideg_R \leq d_k & \end{cases}
\end{equation}
We write $\S_{k| \leq \bd}, \S_{k| < \bd}, \S_{k| \ll \bd} \subset \S_k$ for the vector subspaces of shuffle elements in $k$ variables of slope $ \leq \bd$, $<\bd$, $\ll \bd$, respectively.
\end{definition}

\tab
Now assume we have an infinite sequence of real numbers $(d_1,...,d_k,...)$ such that:
\begin{equation}
\label{eqn:cond}
d_a + d_{a'} \leq d_{a+a'} \qquad \forall \ a,a' \in \BN
\end{equation}
Then in virtue of \eqref{eqn:degree}, we conclude that the vector subspace:
\begin{equation}
\label{eqn:chennai1}
\S_{\leq (d_1,...,d_k,...)} \ := \ \bigoplus_{k=0}^\infty \S_{k| \leq (d_1,...,d_k)}
\end{equation}
is a {\bf subalgebra} of $\S$. Similarly, in virtue of \eqref{eqn:degreeimprove}, we conclude that:
\begin{equation}
\label{eqn:chennai2}
\S_{ < (d_1,...,d_k,...)} \ := \ \bigoplus_{k=0}^\infty \S_{k| < (d_1,...,d_k)}
\end{equation}
\begin{equation}
\label{eqn:chennai3}
\S_{ \ll (d_1,...,d_k,...)} := \bigoplus_{k=0}^\infty \S_{k| \ll (d_1,...,d_k)}
\end{equation}
are both {\bf sub Lie algebras} of $\S$. For example, $d_k = k \cdot \mu$ satisfies \eqref{eqn:cond} for any $\mu \in \BQ$, and this is precisely the choice which was studied in \cite{shuf}. We will write:
$$
\S_{\leq \mu}, \S_{< \mu}, \S_{ \ll \mu} \subset \S
$$
for the subalgebras \eqref{eqn:chennai1}--\eqref{eqn:chennai3} that correspond to the choice $d_k = k \cdot \mu$. \\

\subsection{}\label{sub:filtration}

We will use the filtration of $\S$ by $\S_{\leq \mu}$ to prove Theorem \ref{thm:generation}. The following Lemma is the most important part, and its proof is an adaptation of \cite{F}. \\

\begin{lemma}
\label{lem:magic}

For any vector of real numbers $\bd = (d_1,...,d_k)$, we have: 

\begin{equation}
\label{eqn:dimest}
\dim \ \S_{k|\leq \bd} \ \leq \ \sum_{s\in \BN}\# \Big \{(k_1,e_1),...,(k_s, e_s), \ k_1+...+k_s =k, \ 0 \leq e_i \leq d_{k_i} \Big \} \qquad
\end{equation}
In the right hand side, we count the number of {\bf unordered collections} of $(l,e) \in \BN \times \BN_0$. The analogous count holds for $\S_{k| < \bd}$ (respectively $\S_{k | \ll \bd}$), but with the extra condition that for $s>1$ we only allow $e_i < d_{k_i}$ (respectively $e_i < d_{k_i} - 1$). 

\end{lemma}

\tab
It will follow from Proposition \ref{prop:basis} that the inequalities \eqref{eqn:dimest} are actually equalities. Recall the shuffle elements $C_m$ of \eqref{eqn:ideal}, and define for all $k \in \BN$, $d \in \BN_0$:
\begin{equation}
\label{eqn:pkd}
P_{k,d} = \sym \left [\frac {\prod_{i=1}^k z_i^{\left \lfloor \frac {di}k \right \rfloor - \left \lfloor \frac {d(i-1)}k \right \rfloor} }{\prod_{i=1}^{k-1} (z_{i+1} - z_{i} + t)} \prod_{1 \leq i < j \leq k} \zeta(z_i-z_j) \right]
\end{equation}
$$$$

\begin{proposition}
\label{prop:degree}
For all $k \in \BN$ and $d\in \BN_0$, we have:
$$
P_{k,d} \in \S_{< \frac dk}
$$
$$$$

\end{proposition}

\begin{proposition}
\label{prop:basis}
A linear basis of $\S_{\leq \mu}$ is given by all products:
\begin{equation}
\label{eqn:basis}
P_\Gamma = P_{k_1,d_1} * ... * P_{k_s,d_s}
\end{equation}
over all $s \in \BN$ and all collections: 
\begin{equation}
\label{eqn:collection}
\Gamma = \left\{(k_1,d_1),...,(k_s,d_s), \quad \mu \geq \frac {d_1}{k_1} \geq ... \geq \frac {d_s}{k_s} \right \}
\end{equation}
If $\frac {d_i}{k_i} = \frac {d_j}{k_j}$ for certain $1 \leq i < j \leq s$, then we order according to $k_i \leq k_j$ in \eqref{eqn:collection}.  
\end{proposition}

\tab 
As a corollary, a linear basis of $\S$ is given by the products \eqref{eqn:basis}, with $\mu$ replaced by $\infty$. Since the elements $P_{k,d}$ lie in the subalgebra $\S_\sma$, as a consequence of Proposition \ref{prop:belong}, we conclude that $\S_\sma = \S$. This proves Theorem \ref{thm:generation}. \\

\subsection{}\label{sub:heisvir}

Consider the following particular cases of the shuffle elements \eqref{eqn:pkd}:
\begin{equation}
\label{eqn:b}
\widetilde{B}_{k} \ \ = \ \ \sym \left [\frac 1{\prod_{i=1}^{k-1} (z_{i+1} - z_{i}+t)} \prod_{1 \leq i < j \leq k} \zeta(z_i - z_j) \right] 
\end{equation}
\begin{equation}
\label{eqn:l}
\widetilde{L}_{k} = \sym \left [\frac {z_1+z_k}{2 t_1 t_2 \prod_{i=1}^{k-1} (z_{i+1} - z_{i}+t)} \prod_{1 \leq i < j \leq k} \zeta(z_i - z_j) \right] 
\end{equation}
In Proposition \ref{prop:half}, we will show that the above shuffle elements generate half of a Heisenberg-Virasoro algebra. \\

\begin{proposition}
\label{prop:degheisvir}
For any $k>0$, we have:
\begin{equation}
\label{eqn:moronic}
\widetilde{B}_{k} \in \S_{\ll 0}
\end{equation}
\begin{equation}
\label{eqn:relativism}
\widetilde{L}_{k} \in \S_{\ll \frac 1k}
\end{equation}

\end{proposition}

\tab 
The statements are easy sharpenings of Proposition \ref{prop:degree}, which we leave as exercises to the interested reader. For example, Proposition \ref{prop:degree} establishes the fact that the $l$--degree of $P_{k,0} = \widetilde{B}_k$ is $\leq -1$ for all $l \in \{1,...,k\}$. However, the proof also shows that there are only two summands of $y$--degree equal to $-1$ in the expansion:
$$
P_{k,0} \Big|_{z_1 \mapsto y - t_1,..., z_l \mapsto y - l t_1}
$$
These two summands correspond to $S = \{1,...,l\}$ and $S = \{k-l+1,...,k\}$ in \eqref{eqn:chunk}, and they cancel each other out, thus leaving only terms of degree $\leq - 2$ in $y$. This precisely establishes \eqref{eqn:moronic}. Formula \eqref{eqn:relativism} is proved analogously. \\

\subsection{}\label{sub:identify}

Define the {\bf shadow} of a rational function as:
\begin{equation}
\label{eqn:shadow}
\sh_R(y) = \frac {R\left(y - t_1, y - 2t_1,..., y - k t_1 \right)}{\prod_{1\leq i < j \leq k} \zeta \left(j t_1 - i t_1 \right)}
\end{equation}
Note that the shadow of a shuffle element $R \in \S$ is a polynomial in $y$ of equal to $k\ideg_R$, as defined in the previous Subsection. The following is an easy exercise: \\

\begin{proposition}
\label{prop:shuf}
For any shuffle elements $R,R' \in \S^+$, we have:
\begin{equation}
\label{eqn:sh}
\esh_{R*R'}(y) =  \esh_{R}(y) \cdot \esh_{R'}(y - kt_1)
\end{equation}

\end{proposition}

\tab
We leave the above Proposition to the interested reader, and note that it is proved similarly with Proposition 6.7 of \cite{shuf}. Using the notion of shadow, we can characterize the shuffle elements $\widetilde{B}_{k}$ and $\widetilde{L}_{k}$ implicitly. \\

\begin{proposition}
\label{prop:bound}
For any fixed $k>0$, we have:
\begin{equation}
\label{eqn:dima}
\dim \ \S_{k| \ll 0} \leq 1
\end{equation}
\begin{equation}
\label{eqn:dimb}
\dim \ \S_{k|\ll \frac 1k} \leq 2
\end{equation}
A shuffle element in $\S_{k| \ll \frac 1k}$ is 0 if and only if its shadow is 0.

\end{proposition}

\tab 
Proposition \ref{prop:bound} follows immediately from Lemma \ref{lem:magic}. Indeed, there is a single unordered collection that appears in the right hand side of \eqref{eqn:dimest} for $\S_{k|\ll 0}$, namely $\{(k,0)\}$. Similarly, for $\S_{k|\ll \frac 1k}$ there are only two such unordered collections: $\{(k,0)\}$ and $\{(k,1)\}$. Tracing through the proof of Lemma \ref{lem:magic} shows that the only time we can have a non-zero number in the right hand side of \eqref{eqn:mall} is when $\lambda = (k)$. Then we infer that a shuffle element in $\S_{k| \ll \frac 1k}$ is 0 if and only if it is in the kernel of $\Phi_{(k)}$, which is equivalent with having shadow 0. \\

\begin{proposition}
\label{prop:half}
For any $k,l>0$ we have:
\begin{equation}
\label{eqn:heis}
\left[\widetilde{B}_{k}, \widetilde{B}_{l} \right] = 0
\end{equation}
\begin{equation}
\label{eqn:heisvir}
\left[\widetilde{L}_{k}, \widetilde{B}_{l} \right] = l \widetilde{B}_{k+l}
\end{equation}
\begin{equation}
\label{eqn:vir}
\left[\widetilde{L}_{k}, \widetilde{L}_{l} \right] = (l - k) \widetilde{L}_{k+l}
\end{equation}
Therefore, the elements $\{\widetilde{B}_k, \widetilde{L}_k \}_{k\in \BN}$ generate half of a Heisenberg-Virasoro algebra.
\end{proposition}

\tab 
 Note that formulas \eqref{eqn:heis}--\eqref{eqn:vir} would continue to hold if we transformed:
\begin{equation}
\label{eqn:transformation}
\widetilde{L}_k \ \mapsto \widetilde{L}_k + (k\alpha + \beta) \widetilde{B}_k \qquad \forall \ k>0
\end{equation}
for any constants $\alpha, \beta \in \BF$, which is well-known of the Heisenberg--Virasoro algebra. \\

\subsection{}\label{sub:double}

We will henceforth write $\S^+ = \S$ and $\S^- = \S^{\text{op}}$ (the superscript ``op" refers to the same abelian group, endowed with the opposite ring structure) and refer to these as the {\bf positive} and {\bf negative} shuffle algebras, respectively. Set:
$$
\S^0 = \BF[h_0,h_1,h_2,...]
$$
and let us collect the generators of $\S^0$ into the generating series:
$$
h(w) = \frac {t_1t_2}{-t} + \sum_{k=0}^\infty \frac {h_k}{w^{k+1}}
$$
We define the {\bf double shuffle algebra} as:
\begin{equation}
\label{eqn:double}
\tS = \S^+ \otimes \S^0 \otimes \S^-
\end{equation}
under relations \eqref{eqn:yang1}, \eqref{eqn:yang2} and \eqref{eqn:yang3} below:
\begin{equation}
\label{eqn:yang1}
h(w) * R^+(z_1,...,z_k) \ = \ \left[ R^+(z_1,...,z_k) \prod_{i=1}^k \frac {\zeta(w - z_i)}{\zeta(z_i - w)} \right] * h(w)
\end{equation}
\begin{equation}
\label{eqn:yang2}
R^-(z_1,...,z_{k'}) * h(w) = h(w) * \left[ R^-(z_1,...,z_{k'}) \prod_{i=1}^{k'} \frac {\zeta(w - z_i)}{\zeta(z_i - w)} \right]
\end{equation}
\begin{equation}
\label{eqn:yang3}
\left[ z_{-}^{d}, z_{+}^{d'} \right] = h_{d+d'} \qquad \qquad \forall d,d' \geq 0
\end{equation}
We interpret \eqref{eqn:yang1} and \eqref{eqn:yang2} as collections of equations that arise by expanding in negative powers of $w$ and equating the coefficients in the left and right hand sides.  \\

\begin{proposition}
\label{prop:h0 h1 h2}

The elements $h_0,h_1$ are central in $\tS$, while:
\begin{equation}
\label{eqn:degree operator}
[h_2,R] = \pm k \cdot 2t_1^2t_2^2 \cdot R(z_1,...,z_k)
\end{equation}
\begin{equation}
\label{eqn:degree operator 2}
[h_3,R] = \pm 6t_1^2t_2^2 \cdot (z_1+...+z_k) R \mp 2k t_1 t_2 t \cdot h_0 R
\end{equation}
for all $R \in \S^\pm_k$. Because of \eqref{eqn:degree operator}, we call: 
$$
\frac {h_2}{2t_1^2t_2^2} \in \tS
$$ 
the {\bf degree operator} (up to a constant which we fix in Subsection \ref{sub:explicitformulas}). 

\end{proposition}

\tab
Meanwhile, in \eqref{eqn:yang3} we write $z_{\pm}^d$ for the rational function in one variable $z_1^d$, regarded as an element of the positive/negative shuffle algebra $\S^\pm$. \\ 

\begin{remark} 
\label{rem:l2}

According to Theorem \ref{thm:generation}, any shuffle elements $R^\pm \in \S^\pm$ can be written as sums of products of $z_\pm^d$ for various $d \in \BN$. For example, it is easy to check that the second Heisenberg and Virasoro currents satisfy the formulas:
\begin{equation}
\label{eqn:ll2}
\widetilde{B}_2 := \text{right hand side of \eqref{eqn:b}} = \frac {z^1 * z^0 - z^0 * z^1}{t_1t_2} \in \S
\end{equation}
\begin{equation}
\label{eqn:l2}
\widetilde{L}_2 := \text{right hand side of \eqref{eqn:l}} = \frac {z^2 * z^0 - z^0 * z^2}{2t_1^2t_2^2} \in \S
\end{equation}
Therefore, iterating \eqref{eqn:yang1}--\eqref{eqn:yang3} allows one to express the product $R^- * R^+$ as a sum of elements of $\S^+ * \S^0 * \S^-$. This is the shuffle version of normal ordering. \\

\end{remark}

\subsection{}\label{sub:full}
 
Take the shuffle elements $C_m \in \S$ defined in \eqref{eqn:ideal}, and let us think of them as positive and negative shuffle elements $C_m^\pm \in \S^\pm$. Recall that the positive and negative shuffle algebras are identical as vector spaces, but are endowed with the opposite multiplication. In particular, formulas \eqref{eqn:b} and \eqref{eqn:l} can be interpreted as either positive or negative shuffle elements. We write:
\begin{equation}
\label{eqn:defheis}
B_{-k} = \widetilde{B}_k  \in \S^+, \qquad B_{k} = \widetilde{B}_{k} \in \S^-
\end{equation}
\begin{equation}
\label{eqn:defvir}
L_{\mp k} \ = \ \widetilde{L}_k + \frac {(k-1)h_0 t}{2t_1^2 t_2^2} \cdot \widetilde{B}_k \ \in \ \S^\pm
\end{equation}
for all $k>0$. Furthermore, set:
\begin{equation}
\label{eqn:cartan1}
B_0 = \frac {h_1}{t_1t_2} \qquad \qquad L_0 = \frac {h_2}{2t_1^2t_2^2}
\end{equation}
as well as:
\begin{equation}
\label{eqn:central1}
c_1 = h_0 \qquad \qquad c_2 = h_0 \left (\frac 1{t_1^2} + \frac 1{t_1t_2} + \frac 1{t_2^2} \right) - \frac {h_0^3 t^2}{t_1^4t_2^4} 
\end{equation}
$$$$

\begin{theorem}
\label{thm:full}

The generators $\{B_k, L_k\}_{k\in \BZ}$ induce a Heisenberg-Virasoro algebra:
\begin{equation}
\label{eqn:Heis}
[B_{k},B_{l}] = \delta_{k+l}^0 k \cdot c_1
\end{equation}
\begin{equation}
\label{eqn:Heisvir}
[L_{k},B_{l}] \ = \ - l B_{k+l}  
\end{equation}
\begin{equation}
\label{eqn:Vir}
[L_{k},L_{l}] = (k-l)L_{k+l} + \delta_{k+l}^0 \frac {k(k^2-1)}{12}  \cdot c_2
\end{equation}
where $c_1$ and $c_2$ are central. \\
\end{theorem}

\subsection{} 
\label{sub:connection}

As we close this Section, let us explain the connection between the double shuffle algebra $\tS$ and the Yangian $\Y$ of $\widehat{\fgl}_1$, which was studied in \cite{MO}, \cite{SV1}, \cite{T} and numerous other papers. In a certain incarnation, $\Y$ is generated by the coefficients of power series:
$$
e(w) = \sum_{k=0}^\infty \frac {e_k}{w^{k+1}} \qquad \qquad h(w) = \frac {t_1t_2}{-t} + \sum_{k=0}^\infty \frac {h_k}{w^{k+1}} \qquad \qquad f(w) = \sum_{k=0}^\infty \frac {f_k}{w^{k+1}}
$$
under the relations: 
\begin{equation}
\label{eqn:him1}
[f_d,e_{d'}] = h_{d+d'}
\end{equation}
$$
h(w_1) e(w_2) \cdot (w_1-w_2-t_1)(w_1-w_2-t_2)(w_1-w_2+t) =
$$
\begin{equation}
\label{eqn:him2}
= e(w_2) h(w_1) \cdot (w_1-w_2+t_1)(w_1-w_2+t_2)(w_1-w_2-t) 
\end{equation}
$$
e(w_1) e(w_2) \cdot (w_1-w_2-t_1)(w_1-w_2-t_2)(w_1-w_2+t) =
$$
\begin{equation}
\label{eqn:him3}
= e(w_2) e(w_1) \cdot (w_1-w_2+t_1)(w_1-w_2+t_2)(w_1-w_2-t) 
\end{equation}
as well as the opposite relations for $f$ instead of $e$, and a certain cubic relation (analogous to the Serre relation for Kac-Moody algebras) which we do not recall here. It is straightforward to check that there exists an algebra homomorphism:
$$
\Y \stackrel{\Upsilon}\longrightarrow \tS, \qquad \qquad e_d \mapsto z_+^d, \quad h_d \mapsto h_d, \quad f_d \mapsto z_-^d
$$
simply because relations \eqref{eqn:him1}, \eqref{eqn:him2}, \eqref{eqn:him3} are respected in the shuffle algebra, because of \eqref{eqn:yang3}, \eqref{eqn:yang1}, \eqref{eqn:shufprod}. Theorem \ref{thm:generation} implies that $\Upsilon$ is surjective. Moreover, comparing Proposition \ref{prop:basis} with the dimensions of the filtered pieces of the algebra $\Y$ from \cite{SV1}, we conclude that $\Upsilon$ is an isomorphism. 

\tab 
At this stage, we cannot explain to our reader why the wheel conditions must take the form \eqref{eqn:wheel}, other than by observing that they imply Theorem \ref{thm:generation} and thus the fact that $\Upsilon$ is an isomorphism. In Remark \ref{rem:specialization}, we will see that these conditions are naturally required to insure that the matrix coefficients of the algebra $\tS$ in its level $r$ representation are well-defined, which is a critical fact for our setup. \\

\section{The moduli space of sheaves}
\label{sec:mod}

\subsection{}\label{sub:def} The main geometric object for us is the moduli space $\M_{r,d}$ of rank $r$ degree $d$ torsion-free sheaves on $\BP^2$ which are equipped with a framing:
\begin{equation}
\label{eqn:framing}
\F|_\infty \cong \O^{\oplus r}_\infty
\end{equation}
over a fixed line $\infty \subset \BP^2$. This moduli space is smooth of dimension $2rd$. We consider the action of the torus $T = \BC^* \times \BC^* \times \left( \BC^* \right)^r$ on $\M_{r,d}$, where: \\

\begin{itemize}

\item the first two copies of $\BC^*$ scale $\BP^2$ by keeping the line $\infty$ invariant \\

\item the last $r$ copies of $\BC^*$ act on the trivialization \eqref{eqn:framing}

\end{itemize}

\tab 
We let $t_1,t_2$ be the  equivariant parameters in the direction of $\BC^* \times \BC^*$, and let $u_1,...,u_r$ be the equivariant parameters in the direction of the other $r$ copies of $\BC^*$. Note that the coefficient ring for $T-$equivariant cohomology is $\BZ[t_1,t_2,u_1,...,u_r]$. We abbreviate $\bu = (u_1,...,u_r)$ and consider the cohomology group:
\begin{equation}
\label{eqn:defk}
H_\bu = \bigoplus_{d=0}^\infty H_{\bu, d} \qquad \text{where} \qquad H_{\bu, d} = H_T(\M_{r,d})_\loc
\end{equation}
The subscript refers to localization, i.e. tensoring with $\BF_\bu = \BF(t_1,t_2,u_1,...,u_r)$. \\

\subsection{}\label{sub:taut}

Consider the ring $\Lambda_\bu = \BF_\bu[x_1,x_2,...]^\sym$. For any $d \in \BN$, the moduli space $\M_{r,d}$ admits a rank $d$ {\bf tautological vector bundle} $\V$, with fibers given by:
\begin{equation}
\label{eqn:tautologicalvector}
\V |_\F = H^1(\BP^2, \F(-\infty))
\end{equation}
Let us decompose this vector bundle into formal Chern roots $[\V] = e^{l_1}+...+e^{l_d}$ and define the following homomorphism: 
\begin{equation}
\label{eqn:kirwan}
\Lambda_\bu \longrightarrow \prod_{d=0}^\infty H_{\bu,d} \qquad \qquad f \mapsto \prod_{d=0}^\infty \overline{f}_d
\end{equation}
\begin{equation}
\label{eqn:taut}
\text{where } \qquad \overline{f}_d := f(l_1,...,l_d) \in H_{\bu, d}
\end{equation}
We often abuse notation and denote $\of_d$ simply by $\of$ in situations which do not depend on the number $d \in \BN$. We will abbreviate our sets of variables as:
$$
X = x_1 + x_2 + ... \qquad \text{and} \qquad Z = z_1 + ... + z_k
$$
This notation is well suited for defining {\bf plethysms}, which are homomorphisms:
$$
\Lambda_\bu \longrightarrow \Lambda_\bu, \qquad \qquad f(X) \mapsto f(X \pm Z)
$$ 
completely determined by their image on power sum functions: 
$$
x_1^n + x_2^n + ... \ \mapsto \ x_1^n + x_2^n + ... \pm \left( z_1^n + ... + z_k^n \right)
$$
As a notational rule, we will use the letter $X$ for Chern classes of tautological bundles on $\M_{r,d}$ and the letter $Z$ for the variables of shuffle elements. Expressions:
$$
\zeta(Z-X) := \prod_{i=1}^k \prod_{j=1}^\infty \zeta(z_i-x_j)
$$
should always be considered to be multiplicative in the alphabets of variables $X,Z$. The only exception is when we write $\zeta(Z-Z)$, which is undefined because $\zeta(z - z)$ is an indeterminate. Whenever this happens, we remove the problematic $\zeta$'s:
$$
\zeta(Z-Z) :=  \prod_{1\leq i \neq j \leq k} \zeta(z_i - z_j)
$$

\subsection{}\label{sub:act} 

For any choice of parameters $\bu = (u_1,...,u_r)$, define the polynomial:
\begin{equation}
\label{eqn:tau}
\tau_\bu(z) = (z - u_1)...(z - u_r)
\end{equation}
Whenever we write a formula that involves a choice of sign $\pm$, we set: 
\begin{equation}
\label{eqn:epsilon}
\e = \frac {1\pm 1}2 = \delta_\pm^+ \qquad \qquad \be = \frac {1\mp 1}2 = \delta_\pm^-
\end{equation}
The following Theorem is a shuffle algebra interpretation of the Yangian action on $H_\bu$ that was constructed in \cite{MO}, \cite{SV1}, \cite{T} by various means (also see Section \ref{sub:connection} for the connection between the shuffle algebra and the Yangian). \\

\begin{theorem}
\label{thm:act}

The following formulas give rise to an action $\widetilde{\S} \curvearrowright H_\bu$:
\begin{equation}
\label{eqn:cartan}
h(w) = \text{multiplication by } \frac {t_1t_2}{-t} \cdot \overline{\frac {\zeta(w-X)}{\zeta(X-w)}}_d \cdot \frac {\tau_\bu(w+t)}{\tau_{\bu}(w)}
\end{equation}
on each direct summand $H_{\bu,d}$ of \eqref{eqn:defk}, as well as:
\begin{equation}
\label{eqn:int}
\overline{f}_d \ \stackrel{R^\pm}\mapsto \ \frac {(-1)^{k\be}}{k!} :\int: \frac {R^\pm(Z)}{\zeta(Z-Z)} \cdot \overline{f(X \mp Z) \zeta(\pm Z \mp X)^{\pm 1}}_{d\pm k} \cdot \tau_\bu(Z + t\e)^{\pm 1}
\end{equation}
for any $d\in \BN$ and $R^\pm(Z) = R^\pm(z_1,...,z_k) \in \S^\pm$ (thus the operator $R^\pm$ increases the degree $d$ by $\pm k$). The normal-ordered integral $:\int:$ is defined in Remark \ref{rem:normal order}. \\

\end{theorem}

\begin{remark}
\label{rem:normal order}

When $k=1$ and $R^\pm = z^a$ for $a\in \BN$, formula \eqref{eqn:int} is \textbf{defined} by:
\begin{equation}
\label{eqn:int k=1}
\overline{f}_d \ \mapsto \ (-1)^{\be} \int_{|z| \gg X} z^a \cdot \overline{f(X \mp z) \zeta(\pm z \mp X)^{\pm 1}}_{d\pm 1} \cdot \tau_\bu(z + t\e)^{\pm 1} \frac {dz}{2\pi i}
\end{equation}
where the integral $\int_{|z| \gg X}$ is defined as the residue at $z = \infty$, i.e. it goes over a large contour that surrounds all the $X$ variables and all the parameters $t_1,t_2,\bu$. If we iterate \eqref{eqn:int k=1}, we see that for a positive/negative shuffle element of the form:
\begin{equation}
\label{eqn:shuffle element}
R^\pm(z_1,...,z_k) = z_1^{a_1} * ... * z_k^{a_k} = \sym \left[ z_1^{a_1}...z_k^{a_k} \prod_{1\leq i < j \leq k} \zeta(\pm z_i \mp z_j) \right]
\end{equation}
the fact that \eqref{eqn:int} should give a shuffle algebra action forces us to \textbf{define}:
\begin{equation}
\label{eqn:int general k}
\overline{f}_d \ \stackrel{R^\pm}\mapsto \ (-1)^{k\be} \int_{|z_1| \gg ... \gg |z_k| \gg X} \frac {z_1^{a_1}...z_k^{a_k}}{\prod_{1 \leq i<j \leq k} \zeta (\pm z_j \mp z_i)} \cdot 
\end{equation}
$$
\overline{f(X \mp z_1 \mp ... \mp z_k) \zeta(\pm z_1 \pm ... \pm z_k \mp X)^{\pm 1}}_{d\pm k} \cdot \prod_{i=1}^k \tau_\bu(z_i + t\e)^{\pm 1} \frac {dz_1}{2\pi i}...\frac {dz_k}{2\pi i}
$$
Note that the second line of the above formula is symmetric in the variables $Z = z_1+...+z_k$. While we will not use this fact in the present paper, one can perturb the parameters $t_1,t_2,\bu$ as in \cite{nek} to ensure that \eqref{eqn:int general k} goes over $|z_1| = ... = |z_k| \gg X$. Since the contours would then become symmetric, one can symmetrize the integrand in \eqref{eqn:int general k} without changing the value of the integral:
$$
\text{RHS of \eqref{eqn:int general k}} = \frac {(-1)^{k\be}}{k!} \int_{|z_1| = ... = |z_k| \gg X}  \frac {R^\pm(Z)}{\zeta (Z - Z)} \overline{f(X \mp Z) \zeta(\pm Z \mp X)^{\pm 1}}_{d\pm k}  \tau_\bu(Z + t\e)^{\pm 1}
$$
where the shuffle element $R^\pm$ is the symmetric rational function given by \eqref{eqn:shuffle element}.

\tab 
We will not make the above contour manipulation rigorous, and instead pursue the following approach to define normal-ordered integrals: Theorem \ref{thm:generation} implies that any shuffle element is a linear combination of \eqref{eqn:shuffle element} for various natural numbers $a_1,...,a_k$, so equation \eqref{eqn:int general k} completely defines the normal-ordered integral \eqref{eqn:int}. To ensure consistency, one must check two things: \\

\begin{itemize}

\item The right hand side of \eqref{eqn:int} does not depend on the presentation of $R^\pm$ as a linear combination of shuffle elements \eqref{eqn:shuffle element} \\

\item The right hand side of \eqref{eqn:int} does not depend on the presentation of a cohomology class as $\overline{f}_d$ for some symmetric polynomial $f(X)$  \\

\end{itemize}

\noindent Both of these consistencies will be proved in Proposition \ref{prop:restriction} below. \\

\end{remark}

\subsection{}\label{sub:explicitformulas}

Write $\bar{u}_i = u_i+\frac {(r-1)t}2$ for all $i$. Note that the first few coefficients of \eqref{eqn:cartan} are:
$$
h_0 = - r t_1t_2 \quad \qquad h_1 = - t_1t_2 \sum_{i=1}^r \bar{u}_i 
$$ 
$$
h_2 \ = \ \bd \cdot 2 t^2_1 t^2_2 + \frac {r(r^2-1)t_1 t_2 t^2}{12} - t_1 t_2 \sum_{i=1}^r \bar{u}_i^2 
$$
$$
h_3 = \bc \cdot 6 t^2_1 t^2_2 + \bd \cdot 2 r t^2_1 t^2_2 t - \frac {r^2(r^2-1)t_1 t_2 t^3}{24} +
$$
$$
+ \frac {(r^2-1) t_1 t_2 t^2}4 \sum_{i=1}^r \bar{u}_i + t_1t_2t \left[ \frac {r-1}2 \sum_{i=1}^r \bar{u}_i^2 - \sum_{1\leq i < j \leq r} \bar{u}_i\bar{u}_j \right] - t_1 t_2 \sum_{i=1}^r \bar{u}_i^3 
$$
where $\bc$ is the Calogero-Sutherland Hamiltonian and $\bd$ is the degree operator. In our language, these are given by:
$$
\bc = \text{multiplication by } c_1(\V) \qquad \qquad \qquad \bd = \text{multiplication by rk }\V
$$
In particular, relations \eqref{eqn:cartan1} and \eqref{eqn:central1} imply the following formulas:
\begin{equation}
\label{eqn:cartan2}
L_0 \Big |_{H_{\bu}} \ = \ \bd + \frac {r(r^2-1) t^2}{24t_1t_2} - \frac 1{2t_1t_2} \sum_{i=1}^r \bar{u}_i^2
\end{equation}
\begin{equation}
\label{eqn:central2}
c_1 \Big |_{H_\bu} = - r t_1 t_2 \qquad \qquad B_0 \Big |_{H_\bu} = - \sum_{i=1}^r \bar{u}_i
\end{equation}

\subsection{}\label{sub:fixedpoints1}

The other way to present the action of Theorem \ref{thm:act} is via equivariant localization in the basis of fixed points. Let us recall the description of $T-$fixed points of $\M_{r,d}$. These are indexed by $r${--\bf partitions}, by which we mean collections:
$$
\bla = (\la^1,...,\la^r)
$$
where each $\la^i = (\la^i_1 \geq \la^i_2 \geq ...)$ is an ordinary partition, and the total size: 
$$
|\bla| = |\la^1|+...+|\la^r|
$$
equals $d$. We will use the notation $\bla \vdash d$ to denote $r$--partitions of size $d$. Recall that an ordinary partition $\lambda$ can be identified with its {\bf Young diagram}, which is a set of $1\times 1$ boxes in the first quadrant of the plane:

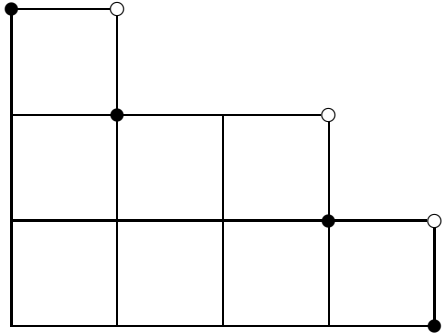
\begin{picture}(110,130)(-105,0)

\put(0,0){\line(1,0){160}}
\put(0,40){\line(1,0){157}}
\put(0,80){\line(1,0){117}}
\put(0,120){\line(1,0){37}}

\put(0,0){\line(0,1){120}}
\put(40,0){\line(0,1){117}}
\put(80,0){\line(0,1){80}}
\put(120,0){\line(0,1){77}}
\put(160,0){\line(0,1){37}}

\put(40,120){\circle{5}}
\put(120,80){\circle{5}}
\put(160,40){\circle{5}}

\put(160,0){\circle*{5}}
\put(120,40){\circle*{5}}
\put(40,80){\circle*{5}}
\put(0,120){\circle*{5}}


\end{picture}

\tab 
For example, the above Young diagram corresponds to the partition $\lambda = (4,3,1)$. The lattice points denoted by black (respectively white) circles are called the {\bf inner} (respectively {\bf outer}) corners of $\lambda$. We will apply the same terminology to an $r$--partition, which consists of $r$ Young diagrams as above. Given a box $\sq = (x,y)$ that belongs to the $i$--th constituent partition $\la^i \subset \bla$, we define its {\bf weight} as:
\begin{equation}
\label{eqn:weight}
\chi_\sq = u_i + xt_1 + yt_2
\end{equation}
It encodes not only the position of the box in the first quadrant of the lattice plane, but also which partition indexed from $1$ to $r$ the box $\sq$ lies in. Note the formula:
\begin{equation}
\label{eqn:formula}
\prod_{\sq \in \bla} \zeta(z-\chi_\sq) \cdot \tau_\bu(z+t) =  \frac {\prod^{\sq \text{ inner}}_{\text{corner of }\bla} (z - \chi_\sq + t)}{\prod^{\sq \text{ outer}}_{\text{corner of }\bla} (z - \chi_\sq + t)}
\end{equation}
for any variable $z$. We will extend all the usual constructions from partitions to $r$--partitions. For example, we write:
$$
\bmu \subset \bla
$$
if $\mu^i \subset \la^i$ for all $i\in \{1,...,r\}$. If this is the case, we call $\blamu$ a {\bf skew} $r$--{\bf partition} and think of it as an ordered collection of $r$ sets of boxes in the first quadrant. \\

\subsection{}\label{sub:shufcoeff} 

The skyscraper sheaves at the torus fixed points have cohomology classes $[\bla]$, which we will renormalize as:
\begin{equation}
\label{eqn:renormalization}
|\bla \rangle = \frac {[\bla]}{e(T_\bla\M_{r,d})} \in H_{\bu,d}
\end{equation}
In this normalization, the Atiyah--Bott {\bf equivariant localization formula} reads:
\begin{equation}
\label{eqn:eqloc}
c = \sum_{\bla \vdash d} c|_\bla \cdot |\bla \rangle
\end{equation}
for any class $c \in H_{\bu,d}$. For us, it will be very important to study the restrictions of the tautological vector bundle to the torus fixed points:
\begin{equation}
\label{eqn:tautrest0}
\V|_\bla = \sum_{\sq \in \bla} e^{\chi_\sq} 
\end{equation}
Then for any symmetric polynomial $f \in \Lambda_\bu$, we see that:
\begin{equation}
\label{eqn:tautrest}
\of_d|_\bla = f(\bla) := f(...,\chi_\sq,...)_{\sq \in \bla}
\end{equation}
In particular, we see that $\of_d = 0$ iff $f(\bla) = 0$ for all $r$--partitions $\bla \vdash d$. This shows that the consistency check in the second bullet of Remark \ref{rem:normal order} is non-trivial: there are many ways to represent a cohomology class as $\of_d$ for a symmetric polynomial $f$. To perform this check, we will express formula \eqref{eqn:int} in the basis $|\bla \rangle$: \\

\begin{proposition}
\label{prop:restriction}
For any shuffle element $R^\pm \in \S^\pm$, its matrix coefficients in the basis of renormalized fixed points are given by:
\begin{equation}
\label{eqn:coeff+}
\langle \bla | R^+ | \bmu \rangle = R^+(\blamu) \prod_{\bsq \in \blamu} \left[\frac {t_1t_2}t \prod_{\sq \in \bmu} \zeta(\chi_\bsq - \chi_\sq) \tau_\bu(\bsq+t) \right]
\end{equation}
\begin{equation}
\label{eqn:coeff-}
\langle \bmu | R^- | \bla \rangle = R^-(\blamu) \prod_{\bsq \in \blamu} \left[\frac {t_1t_2}t \prod_{\sq \in \bla} \zeta(\chi_\sq - \chi_\bsq)^{-1} \frac {1}{\tau_\bu(\bsq)} \right]
\end{equation}
where $R^\pm(\blamu) = R^\pm(...,\chi_\sq,...)_{\sq \in \blamu}$. If $\bmu \not \subset \bla$, set $\langle \bla | R^+ | \bmu \rangle = \langle \bmu | R^- | \bla \rangle = 0$. \\

\end{proposition}

\begin{remark}
\label{rem:specialization}

Since shuffle elements are symmetric rational functions of the form:
$$
R(z_1,...,z_k) = \frac {\rho(z_1,...,z_k)}{\prod_{1\leq i \neq j \leq k} (z_i - z_j+t)}
$$
we must explain why the evaluations $R(\blamu)$ of Proposition \ref{prop:restriction} are well-defined. This will be strongly contingent on the fact that the numerator $\rho$ satisfies the wheel conditions \eqref{eqn:wheel} and that $\blamu$ is a skew $r$--partition. Specifically, consider a box $\bsq \in \blamu$ situated in an outer corner (i.e. such that $\bnu = \bla \backslash \bsq$ is an $r$--partition):
$$
\blamu = \bnumu \sqcup \bsq
$$
Then we set:
$$
R(\blamu) = \frac {\rho\left( \{\chi_\sq\}_{\sq \in \bnumu}, z \right)}{\prod_{\sq \neq \sq' \in \bnumu} (\chi_\sq - \chi_{\sq'}-t) \cdot \prod_{\sq \in \bnumu} (z - \chi_\sq - t)(\chi_{\sq} - z - t)} \Big|_{z\mapsto \chi_\bsq}
$$
which is well-defined for the following reason: if the denominator blows up at the evaluation $z\mapsto \chi_\bsq$, it can only be because of a pole of the form $z = \chi_\sq + t$ where $\sq$ is the box situated directly southwest of $\bsq$. In this case, the box $\sq$ belongs to the skew $r$--partition $\blamu$, hence the same must be true of the box $\sq_1$ directly west of $\bsq$ and the box $\sq_2$ directly south of $\bsq$. Therefore, the wheel conditions \eqref{eqn:wheel} imply that $\rho$ has a zero at $z = \chi_\bsq$, and this precisely cancels out the pole. \\

\end{remark}

\subsection{}\label{sub:syt}

The gist of Proposition \ref{prop:restriction} is that, up to some predictable linear factors, the matrix coefficients of the operators $R^\pm \curvearrowright H_\bu$ are given by evaluating these shuffle elements at the set of boxes in a skew $r$--partition $\blamu$. When:
$$
R^\pm = C_m^\pm
$$
are the shuffle elements of \eqref{eqn:ideal}, for any polynomial $m \in \BF_\bu[z_1,...,z_k]$, taking the evaluation $C_m^\pm(\blamu)$ corresponds to all ways of labeling the boxes:
\begin{equation}
\label{eqn:labeling}
\Big \{\sq \in \blamu \Big\} = \Big\{ \sq_1,...,\sq_k \Big\}
\end{equation}
In other words, for every such labeling, we need to plug $z_i \mapsto \chi_i := \chi_{\sq_i}$ in formula \eqref{eqn:ideal}. Because $\zeta(-t_1) = \zeta(-t_2) = 0$, the only labelings which produce non-zero terms are those for which the box $\sq_i$ is not one unit below or left of the box $\sq_j$, for any $i<j$. With this in mind, we recall the following definition: \\

\begin{definition}
\label{def:asyt}

A {\bf standard Young tableau} of shape $\blamu$, abbreviated SYT, is a labeling \eqref{eqn:labeling} such that the labels decrease as we go up and to the right in each of the $r$ constituent partitions of $\blamu$.

\end{definition}

\tab 
Equivalently, a SYT of shape $\blamu$ can be represented as a flag of partitions: 
\begin{equation}
\label{eqn:flagsyt}
\bmu = \bnu_k \subset \bnu_{k-1} \subset ... \subset \bnu_1 \subset \bnu_0 = \bla
\end{equation}
where $|\bnu_{i-1} \backslash \bnu_i| = 1$ for all $i$. We obtain the following Corollary of Proposition \ref{prop:restriction}: \\

\begin{corollary}
\label{cor:restriction}
For the positive/negative shuffle elements $C_m^\pm$ of \eqref{eqn:ideal}, we have:
$$
\langle \bla | C_m^+ | \bmu \rangle = \sum^{\emph{SYT of}}_{\emph{shape }\blamu} \frac {m(\chi_1,...,\chi_k) \prod_{i < j} \zeta(\chi_i - \chi_j)}{\prod_{i=1}^{k-1} (\chi_{i+1} - \chi_{i}+t)}  \prod_{i=1}^k \left[ \frac {t_1t_2}t \prod_{\sq \in \bmu} \zeta(\chi_i - \chi_\sq) \tau_\bu(\chi_i+t) \right]
$$
$$
\langle \bmu | C_m^- | \bla \rangle \ = \ \sum^{\emph{SYT of}}_{\emph{shape }\blamu} \frac {m(\chi_1,...,\chi_k) \prod_{i < j} \zeta(\chi_i - \chi_j)}{\prod_{i=1}^{k-1} (\chi_{i+1} - \chi_{i}+t)}\prod_{i=1}^k \left[ \frac {t_1t_2}t \prod_{\sq \in \bla} \zeta(\chi_\sq - \chi_i)^{-1} \frac {1}{\tau_\bu(\chi_i)} \right]
$$
where $\chi_1,...,\chi_k$ denote the weights of the boxes labelled $1,...,k$ in a $\esyt$. \\
\end{corollary}

\subsection{}\label{sub:flag}

The formulas in Corollary \ref{cor:restriction} give hints as to which kind of geometric correspondence may give rise to the action of $C_m^\pm$ on $H_\bu$, since the fixed points of such a correspondence should be indexed by standard Young tableaux. To pursue this idea, recall the ADHM description of the moduli space of framed sheaves: \\

\begin{theorem}
\label{thm:nak}
{\bf (\cite{Nak})} The variety $\M_{r,d}$ is isomorphic to the space of quadruples:

\begin{equation}
\label{eqn:adhm}
(X,Y,A,B) \in \emph{End}(\BC^d) \times \emph{End}(\BC^d) \times \emph{Hom}(\BC^r,\BC^d) \times \emph{Hom}(\BC^d,\BC^r) 
\end{equation}
satisfying the closed condition:
$$
\mu (X,Y,A,B) := [X,Y]+AB = 0 \in \emph{End}(\BC^d)
$$
the open condition that $\BC^d$ is generated by $X,Y$ acting on $\emph{Im }A$, and taken modulo the action of $GL_d$ by conjugation: $g\cdot (X,Y,A,B) = (gXg^{-1}, gYg^{-1},gA,Bg^{-1})$.
\end{theorem}

\tab 
The above allows us to compute the $K$--theory class of the tangent space to $\M_{r,d}$:
\begin{equation}
\label{eqn:tangentclass}
\left[T \M_{r,d} \right] = \sum_{i=1}^r \left( \frac {\V}{e^{u_i}} + \frac {e^{u_i-t}}{\V} \right) - \left(1 - \frac 1{e^{t_1}} \right)\left(1 - \frac 1{e^{t_2}} \right)\frac {\V}{\V}
\end{equation}
where we abuse notation and write $\V$ for the $K$--theory class of the tautological vector bundle \eqref{eqn:tautologicalvector}. Here and throughout this paper, we use the notation:
$$
\frac {\V'}{\V} \quad \text{instead of} \quad [\V'] \otimes [\V^\vee]
$$
for any vector bundles $\V,\V'$. Formula \eqref{eqn:tangentclass} is the special case $d_+ = d_-$ of Proposition \ref{prop:virtualtangent} below, but let us sketch its proof in order to see the motivation behind it. The description of the moduli space $\M_{r,d}$ as the set of certain quadruples \eqref{eqn:adhm} allows us to write its tangent space as:

\begin{equation}
\label{eqn:master}
\left[T \M_{r,d} \right] = \Big[\text{affine space of }X,Y,A,B \Big] - \Big[\text{equation } \mu = 0 \Big] - \Big[ \text{Lie } GL_d \Big] \qquad
\end{equation}
The contributions to \eqref{eqn:master} of the affine spaces of matrix entries of $X,Y,A,B$ are precisely $e^{-t_1} \V \otimes \V^\vee$, $e^{-t_2} \V \otimes \V^\vee$, $\sum_{i=1}^r e^{-u_i} \V$ and $\sum_{i=1}^r e^{u_i-t} \V^\vee$, respectively. The appearance of the equivariant parameters is due to the fact that $X$ and $Y$ are scaled by the rank 2 torus with equivariant parameters $t_1,t_2$, while the directions of $\BC^r$ are scaled by the rank $r$ torus with equivariant parameters $u_1,...,u_r$. The contribution to \eqref{eqn:master} of the equation $\mu = 0$ is $e^{-t}\V \otimes \V^\vee$, and the contribution of the gauge group $GL_d$ is $\V \otimes \V^\vee$. Adding and subtracting all of these contributions according to \eqref{eqn:master} gives us \eqref{eqn:tangentclass}. \\

\subsection{}\label{sub:correspondences} We will always refer to $(X,Y,A,B)_d$ as a quadruple of the form \eqref{eqn:adhm}, where the subscript keeps track of the dimension of the vector space $\BC^d$. For any pair of natural numbers $d_+ > d_-$, let us fix a quotient of vector spaces:
\begin{equation}
\label{eqn:quotient}
\BC^{d_+} \twoheadrightarrow \BC^{d_-}
\end{equation}
Consider the correspondence $\fC_{d_+,d_-} \subset \M_{r,d_+} \times \M_{r,d_-}$ consisting of pairs of sheaves $(\F_+,\F_-)$ such that $\F_+ \subset \F_-$. In the ADHM picture, it can be thought of as:
\begin{equation}
\label{eqn:coarse}
\fC_{d_+,d_-} = \Big \{ (X,Y,A,B)_{d_+} \text{which preserve \eqref{eqn:quotient}} \Big\}/P_{d_+,d_-}
\end{equation}
where $P_{d_+,d_-} \subset GL_{d_+}$ is the subgroup of automorphisms that preserve the quotient \eqref{eqn:quotient}. A variant of this construction was introduced by Baranovsky in \cite{B}, who studied the locus $\fB_{d_+,d_-} \subset \fC_{d_+,d_-}$ of pairs of sheaves such that $\F_+ \subset \F_-$ and the quotient $\F_-/\F_+$ is supported at the origin. In the ADHM picture, this variety is:
$$
\fB_{d_+,d_-} = \left\{\begin{array}{ll} (X,Y,A,B)_{d_+} \text{ which preserve \eqref{eqn:quotient} and} \\
\text{ } X,Y \text{ are nilpotent on } \text{Ker}(\BC^{d_+} \twoheadrightarrow \BC^{d_-}) \end{array} \right\} / P_{d_+,d_-}
$$
In \cite{mod}, we defined certain correspondences that refined $\fB_{d_+,d_-}$ by introducing a full flag of sheaves between $\F_-$ and $\F_+$. In the ADHM language, we fix a full flag:
\begin{equation}
\label{eqn:flag}
0 = V_0 \subset V_{1} \subset ... \subset V_{k-1} \subset V_{k} = \text{Ker } \Big( \BC^{d_+} \twoheadrightarrow \BC^{d_-} \Big)
\end{equation}
of vector spaces, where $k = d_+ - d_-$, and give the following definition. \\

\begin{definition}
\label{def:fine}
The {\bf fine correspondence} $\fZ_{d_+,d_-}$ parametrizes full flags:
\begin{equation}
\label{eqn:fine}
\F_+ = \F_0 \subset \F_{1} \subset ... \subset \F_{k-1} \subset \F_k = \F_-
\end{equation}
of framed sheaves on $\BP^2$, such that the successive quotients $\F_{i}/\F_{i-1}$ are length 1 skyscraper sheaves supported at the origin. In the ADHM picture, this reads: 
$$
\fZ_{d_+,d_-} = \Big \{ (X,Y,A,B)_{d_+} \text{ which preserve \eqref{eqn:flag} and }X,Y\text{ nilpotent on }V_k \Big \} / B_{d_+,d_-}
$$
where $B_{d_+,d_-} \subset P_{d_+,d_-}$ is the subgroup of automorphisms which preserve \eqref{eqn:flag}. \\

\end{definition}

\subsection{} \label{sub:virtual}

The variety $\fZ_{d_+,d_-}$ is quite badly behaved, so will use the ADHM picture to define a {\bf virtual fundamental class}:
\begin{equation}
\label{eqn:virtual}
\left[ \fZ^\vir_{d_+,d_-} \right] \ \in \ H_T \left( \fZ_{d_+,d_-} \right)
\end{equation}
This is done by taking the fundamental class of the affine space of linear maps $X,Y,A,B$ as in Definition \ref{def:fine}, and considering the cohomology class cut out by the equations $[X,Y]+AB = 0$. This class is equivariant under $B_{d_+,d_-}$, so it descends to a cohomology class on the quotient $\fZ_{d_+,d_-}$. This will be the virtual fundamental class defined in \eqref{eqn:virtual}. Moreover, the presentation by generators and relations allows us to define the {\bf virtual tangent space}:
\begin{equation}
\label{eqn:virtualtangent}
\left[ T^\vir \fZ_{d_+,d_-} \right] \ \in \ K_T \left( \fZ_{d_+,d_-} \right)
\end{equation}
By analogy with \eqref{eqn:master}, we may express this $K$--theory class as:
\begin{equation}
\label{eqn:puppets}
\left[  T^\vir \fZ_{d_+,d_-} \right] = \left[\begin{array}{ll} \text{affine space of }X,Y,A,B \\ \ \text{ as in Definition \ref{def:fine}} \end{array} \right]  - 
\end{equation}
$$
- \left[\begin{array}{ll} \text{vector space where }d\mu \text{ takes values,} \\ \text{where }\mu(X,Y,A,B)= [X,Y]+AB\end{array} \right] - \Big[ \text{Lie } B_{d_+,d_-} \Big]
$$
To make the above formula useful, we must express the $K$--theory classes of the affine spaces that appear in \eqref{eqn:puppets} in terms of tautological classes on $\fZ_{d_+,d_-}$. \\

\subsection{}\label{sec:operators}

Consider the projection maps that forget all but the first/last sheaf in \eqref{eqn:fine}:
$$
\xymatrix{
& \fZ_{d_+,d_-} \ar[rd]^{\pi^-} \ar[ld]_{\pi^+} & \\
\M_{r,d_+}  & & \M_{r,d_-}}
$$
These allow us to pull-back tautological vector bundles \eqref{eqn:tautologicalvector} from $\M_{r,d_\pm}$. We write $\V^\pm = \pi^{\pm *}(\V)$ for the resulting bundles on $\fZ_{d_+,d_-}$. Moreover, we have line bundles: 
$$
\L_1,...,\L_{k} \in \text{Pic}_T \left( \fZ_{d_+,d_-} \right)
$$
where $k = d_+ - d_-$ and $\L_i$ is induced by the $i-$th elementary character of the subgroup $B_{d_+,d_-}$. In other words, $\L_i$ keeps track of the $i$--th (from left to right) inclusion in the flag \eqref{eqn:flag}. Note the following equality of $K$--theory classes:
$$
\V^+ = \V^- + \L_1 + ... + \L_k  \ \in \ K_T \left( \fZ_{d_+,d_-} \right)
$$

\begin{proposition}
\label{prop:virtualtangent}
We have the following equality in equivariant $K$--theory:
\begin{equation}
\label{eqn:tangentclass1}
\left[ T^\evir \ \fZ_{d_+,d_-}\right]  = \frac k{e^{t}} - \frac k{e^{t_1}} - \frac k{e^{t_2}} + \sum_{i=1}^{k-1} \frac {\L_{i}}{e^{t}\L_{i+1}}  -
\end{equation}
$$
- \left(1 - \frac 1{e^{t_1}} \right)\left(1 - \frac 1{e^{t_2}} \right)\left(\frac {\V_+}{\V_-} + \sum_{1\leq j \leq i \leq k} \frac {\L_j}{\L_i} \right) +  \sum_{i=1}^r \left(\frac {\V_+}{e^{u_i}} + \frac {e^{u_i-t}}{\V_-} \right)
$$
\end{proposition}

\tab
For all $i\in \{1,...,k\}$ where $k=d_+-d_-$, let us write:
$$
l_i = c_1(\L_i) \ \in \ H_T \left( \fZ_{d_+,d_-} \right)
$$
Recall the virtual tangent space $\left[ T^\vir \ \fZ_{d_+,d_-}\right]$ of \eqref{eqn:virtualtangent}, and let us define operators:

\begin{equation}
\label{eqn:op+}
H_{\bu, \bullet} \stackrel{x_m^+}\longrightarrow H_{\bu, \bullet + k}, \qquad \qquad c \mapsto - \pi^+_* \Big(\left[ \fZ^\vir_{\bullet+k,\bullet} \right]  m\left(l_1,...,l_k \right) \cdot \pi^{-*}(c) \Big) \qquad
\end{equation}
\begin{equation}
\label{eqn:op-}
H_{\bu, \bullet} \stackrel{x_m^-}\longrightarrow H_{\bu, \bullet - k}, \quad c \mapsto  (-1)^{kr-1} \pi^-_* \Big(\left[ \fZ^\vir_{\bullet,\bullet - k} \right] m\left(l_1,...,l_k \right) \cdot \pi^{+*} (c) \Big) \qquad
\end{equation}
for any $m \in \BF_\bu[z_1,...,z_k]$. More rigorously, the push-forwards $\pi^\pm_* \left( [ \fZ^\vir_{\bullet,\bullet'}] \cdot ... \right)$ are defined with respect to the virtual tangent class of \eqref{eqn:virtualtangent}, \eqref{eqn:puppets}, \eqref{eqn:tangentclass1}. \\

\subsection{}\label{sub:fixedpoints2}

Let us now describe the fixed points of $\M_{r,d}$ in terms of the ADHM presentation of the moduli spaces of sheaves. Given an $r$--partition $\bla \vdash d$, we set:
\begin{equation}
\label{eqn:quirky}
\BC^d = \bigoplus_{\sq \in \bla} \BC_{\sq}
\end{equation}
and construct a quadruple by letting: 
$$
\begin{cases} X: \BC_\sq \mapsto \BC_{\text{box one unit to the right of }\sq} \qquad \quad Y: \BC_\sq \mapsto \BC_{\text{box one unit above }\sq} \\ 
A: \left( i\text{--th basis vector of }\BC^r \right) \mapsto \BC_{\text{southwest corner of }\la^{i} \subset \bla} \qquad \qquad \quad \ B = 0 \end{cases} 
$$
We can provide a similar description for the torus fixed points of the correspondences in Subsection \ref{sub:flag}. More specifically, a fixed point of the correspondence $\fC_{d_+,d_-}$ of \eqref{eqn:coarse}, or equivalently a fixed point of $\fB_{d_+,d_-}$, is a pair of $r$--partitions:
$$
(\bla_+, \bla_-) \qquad \text{such that} \qquad \bla_+ \supset \bla_-
$$
A fixed point of the fine correspondence $\fZ_{d_+,d_-}$ of \eqref{eqn:fine} consists of a pair of $r$--partitions as above, but the fact that we work over a fixed flag \eqref{eqn:flag} means that we have a labeling of the boxes of $\bla_+ \backslash \bla_-$:
\begin{equation}
\label{eqn:label}
\bla_+ \backslash \bla_- = \Big\{ \sq_1,...,\sq_{k} \Big\}
\end{equation}
where $k = d_+-d_-$. In the description \eqref{eqn:quirky}, the one-dimensional space corresponding to the box $\sq_i$ corresponds to the one-dimensional quotient $V_{i}/V_{i-1}$ in \eqref{eqn:flag}. The flag preservation condition requires us to have $X,Y:V_{i} \mapsto V_{i}$ for all $i$, and implies the fact the box $\sq_{i}$ cannot be one unit below or left of $\sq_{j}$, for any $i<j$. In other words, \eqref{eqn:label} gives rise to a standard Young tableau. To summarize:
\begin{equation}
\label{eqn:fixedfine}
\left( \fZ_{d_+, d_-} \right)^T = \Big \{ \text{standard Young tableaux} \Big \}
\end{equation}
where the shape $\blamu$ of the standard Young tableau satisfies $|\bla| = d_+$, $|\bmu|  = d_-$. \\

\subsection{}\label{sub:formulas}

The description of fixed points of $\fZ_{d_+,d_-}$ as SYTx allows us to compare the geometric operators $x_m^\pm$ with the shuffle elements $C_m^\pm$ from Corollary \ref{cor:restriction}. \\

\begin{proposition}
\label{prop:compare}
The operators $x_m^\pm$ act on $H_\bu$ as the shuffle elements:
$$
C^\pm_m =  \esym \left [\frac {m(z_1,...,z_k)}{\prod_{i=1}^{k-1} (z_{i+1} - z_{i}+t)} \prod_{1 \leq i < j \leq k} \zeta(z_i-z_j) \right] \in \S^\pm
$$
$$$$
\end{proposition}

\begin{remark}

As a consequence of Proposition \ref{prop:compare}, the particular class: 
$$
m(z_1,...,z_k) = \frac {z_1+z_k}{2t_1 t_2} + \frac {(k-1) h_0 t}{2t_1^2 t_2^2}
$$
on the fine correspondences $\fZ_{\bullet+k,\bullet}$ gives rise to the Virasoro generators $L_{\mp k}$ acting on $H_\bu$. When $m(z_1,...,z_k)=1$, the Proposition implies that fine correspondences give rise to the Heisenberg generators $B_{\mp k}$, which were already constructed by Baranovsky (\cite{B}) using the correspondences $\fB_{d_+,d_-}$. The fact that our operators coincide with those of \loccit imply the following equality of cohomology classes:
$$
\rho_* \left([\fZ^\vir_{d_+,d_-}] \right) = [\fB_{d_+,d_-}] 
$$ 
where $\rho: \fZ_{d_+,d_-} \rightarrow \fB_{d_+,d_-}$ forgets the intermediate sheaves in the flag \eqref{eqn:fine}. \\

\end{remark}

\noindent 
The proof of Proposition \ref{prop:compare} will be given in the Appendix, when we compare the matrix coefficients $\langle \bla|x_m^\pm |\bmu \rangle$ with those prescribed by Corollary \ref{cor:restriction}. The following formula realizes these coefficients as a certain residue computation, which will be used in proving Theorem \ref{thm:main} (compare these formulas with \eqref{eqn:int general k}). \\

\begin{proposition}
\label{prop:imp}

Let $k = d_+ - d_-$. Consider any polynomial $m(l_1,...,l_k) \in H_T ( \fZ_{d_+,d_-})$ with coefficients pulled back from $\M_{r,d_\mp}$ via $\pi^{\mp *}$. Then we have:
\begin{equation}
\label{eqn:imp}
\pi^{\pm}_* \left( \left[ \fZ^\evir_{d_+,d_-} \right] m(l_1,..,l_k) \right) = - \int^\pm  \frac {dz_1}{2\pi i} \ ... \ \frac {dz_k}{2\pi i}
\end{equation}
$$
\frac {m(z_1,...,z_k) \prod_{i<j} \zeta(z_j-z_i)^{-1}}{\prod_{i=1}^{k-1} (z_{i+1}-z_{i}+t)} \prod_{i=1}^k \left[ \overline{\zeta(\pm z_i \mp X)^{\pm 1}} \cdot \frac {\tau_\bu(z_i+t\e)^{\pm 1}}{(-1)^{r\be}} \right]
$$
where we define:
\begin{equation}
\label{eqn:normal3}
\int^+ = \int_{|z_k| \gg ... \gg |z_1| \gg X} \qquad \int^- = \int_{|z_1| \gg ... \gg |z_k| \gg X}
\end{equation}
Recall that $\e$ is 1 or 0, and $\be$ is 0 or 1, depending on whether the sign is $+$ or $-$. \\

\begin{remark}
\label{rem:other side}

In all our integrals, the parameters $t_1$ and $t_2$ have very small absolute value, specifically much smaller than the distance between the contours in \eqref{eqn:normal3}. On the other hand, the $X$ variables (as well as the parameters $\bu$) are formal symbols which can be specialized to any complex numbers one sees fit. In particular, \eqref{eqn:normal3} assumes these formal symbols to be ``smaller" than the $z$ variables. The alternative, which amount to thinking that the $X$ variables (as well as the parameters $\bu$) are ``bigger" than the $z$ variable, states that \eqref{eqn:imp} also holds if we define:
\begin{equation}
\label{eqn:normal4}
\int^{+} =  \int^\wi_{X \gg |z_1| \gg ... \gg |z_k|} \qquad \int^{-} =  \int^\wi_{X \gg |z_k| \gg ... \gg |z_1|}
\end{equation}
The phrase ``without $\infty$" means that we must remove the residues at $\infty$, so each variable $z_i$ is integrated not over a single circle of radius $|z_i|$, but over the difference between that circle and an auxiliary circle that surrounds $\infty$ and nothing else. The proof of \eqref{eqn:normal4} is identical to that of \eqref{eqn:normal3}, since in both cases we successively compute $k$--fold residues at the $X$ variables. \\

\end{remark}

\end{proposition}

\subsection{}\label{sub:ext}

For any collection of parameters $\bu = (u_1,...,u_r)$, we will write: 
$$
\M_\bu = \bigsqcup_{d=0}^\infty \M_{r,d}
$$
if we wish to emphasize the equivariant parameters of the rank $r$ torus action on this moduli space. Recall the Ext bundle of \eqref{eqn:ext}:
\begin{equation}
\label{eqn:e}
\xymatrix{& \E \ar@{-->}[d] \\
& \M_{\bu} \times \M_{\bu'}}
\end{equation}
defined with respect to two sets of equivariant parameters $\bu = (u_1,...,u_r)$ and $\bu' = (u_1',...,u_r')$. The $K$--theory class of this vector bundle is given by:
\begin{equation}
\label{eqn:eformula}
\left[ \E \right] = \sum_{i=1}^r \left( \frac {\V}{e^{u'_i}} + \frac {e^{u_i-t}}{\V'} \right) - \left(1 - \frac 1{e^{t_1}} \right)\left(1 - \frac 1{e^{t_2}} \right)\frac {\V}{\V'}
\end{equation}
where $\V$ and $\V'$ are pull-backs of the tautological vector bundles from the two factors of \eqref{eqn:e}. Comparing this formula with \eqref{eqn:tangentclass} allows one to check the fact that $\E|_{\text{diagonal}} \cong T \M_{\bu}$. We normalize the Chern polynomial as: 
$$
c(\E,m) = (-1)^{r \cdot \text{rank }\V'}\sum_{i=0}^{\text{rank }\E} c_i(\E) \cdot m^{\text{rank }\E-i} \ \in \ H_\bu \otimes H_{\bu'}
$$
and formula \eqref{eqn:eformula} can be rewritten as:

\begin{equation}
\label{eqn:chern}
c(\E,m) = \overline{\zeta(X' - X - m) \tau_{\bu'}(X+m) \tau_{\bu}(X' - m + t)} \in H_\bu \otimes H_{\bu'}
\end{equation}
where $X$ and $X'$ are place-holders for tautological classes on the two factors of \eqref{eqn:e}. Our main operator $A_m = A_m(1) : H_{\bu'} \rightarrow H_\bu$ is induced by the class \eqref{eqn:chern} when used as a correspondence between $\M_{\bu}$ and $\M_{\bu'}$, as in \eqref{eqn:extoperator}. \\

\subsection{}\label{sub:m0}

For any pair of framed sheaves $\F$ and $\F'$, consider the long exact sequence:
$$
... \longrightarrow \Hom(\F',\F) \longrightarrow \Hom(\F',\F|_\infty) \stackrel{\delta}\longrightarrow \Ext^1(\F',\F(-\infty)) \longrightarrow ...
$$
and consider the canonical element $\K_{\F,\F'}$ in the middle space that comes from projection followed by framing: $\F' \twoheadrightarrow \F'|_\infty \cong \F|_\infty$. The bundle $\E$ has a section:
$$
s |_{(\F,\F')} = \delta\left(\K_{\F,\F'}\right)
$$
which vanishes if and only if $\F' \subset \F$. However, this section has the correct equivariance only if we specialize $\bu = \bu'$. If this is the case, then the existence of this section implies that the operator $A_0|_{\bu = \bu'}$ is given by a correspondence supported on the locus $\{\F' \subset \F\}$. Therefore, we have:
$$
A_0\Big|_{\bu = \bu'} = \sum_{k=0}^\infty g_k \qquad \text{where} \qquad g_k:H_{\bu, \bullet} \rightarrow H_{\bu, \bullet-k}
$$
and $g_k$ is a correspondence supported on the locus $\{\F' \subset \F, \text{length }\F/\F' = k\}$. The operator $g_k$ was identfied in \cite{mod} with the action of the constant shuffle element:
$$
G_k(z_1,...,z_k) := \frac {t^k}{t_1^k t_2^k} \in \S^-_k
$$
This was achieved by comparing the matrix coefficients $\langle \bmu | g_k | \bla \rangle$, computed via \eqref{eqn:chern}, with the matrix coefficients of $\langle \bmu | G_k | \bla \rangle$, computed via Proposition \ref{prop:restriction}. \\

\begin{proposition}
\label{prop:a0}
In any rank $r \geq 1$, we have:
\begin{equation}
\label{eqn:a0}
A_0 \Big|_{\bu = \bu'} = \exp \left(\frac t{t_1t_2} \sum_{k=1}^\infty \frac {B_{k}}k \right) 
\end{equation}
\end{proposition}

\tab 
Proposition \ref{prop:a0} is a purely algebraic statement, which follows from degenerating formula (6.9) of \cite{mod} from the trigonometric to the rational case. This implies that $\Omega_0|_{\bu = \bu'} = 1$ in any rank $r$, where $\Omega_m = \Omega_m(1)$ is defined by \eqref{eqn:factor}. \\

\subsection{}\label{sub:proof}

Using Proposition \ref{prop:imp} and \eqref{eqn:chern}, we will now prove our main Theorem: \\

\begin{proof}{\bf of Theorem \ref{thm:main}:} Let us translate formulas \eqref{eqn:Inter1}--\eqref{eqn:Inter2} into equalities of cohomology classes on $H_\bu \otimes H_{\bu'}$. For all $k>0$, we will write:
\begin{equation}
\label{eqn:notation}
\alpha_{\pm k}, \ \beta_{\pm k}, \ \walpha_{\pm k}, \ \wbeta_{\pm k}, \ \wgamma_{\pm k} \in H_{\bu} \otimes H_{\bu'}
\end{equation}
for those classes which, when multiplied by $c(\E,m)$, give rise to the correspondences: 
$$
A_m \circ B_{\pm k}, \ B_{\pm k} \circ A_m, \ A_m \circ L_{\pm k}, \ L_{\pm k} \circ A_m, \ \left(L_{\pm k} + \frac {m-t\e}{t_1t_2} B_{\pm k} \right) \circ A_m
$$
respectively. We recall that $\e = \delta_\pm^+$ and $\be = \delta_\pm^-$ were defined in \eqref{eqn:epsilon}. When $k=1$, formulas \eqref{eqn:Inter1}--\eqref{eqn:Inter2} reduce to the following equalities in $H_{\bu} \otimes H_{\bu'}$:
\begin{equation}
\label{eqn:ww1}
\beta_{\pm 1} - \alpha_{\pm 1} =  |\barbu'| - |\barbu| - r(m - t\be)
\end{equation}
\begin{equation}
\label{eqn:ww2}
\wgamma_{\pm 1} - \walpha_{\pm 1} = \widetilde{\bd} - \widetilde{\bd}' \mp
\end{equation}
$$
\mp \left[\frac {r(r\pm 1)m(m-t)}2 + \frac {(r\pm 1)(2m-t)}2 |\barbu| - \frac {2rm - (r+1)t}2 |\barbu'| + \frac {(|\barbu| - |\barbu'|)^2}2\right] 
$$
When $k>1$, formulas \eqref{eqn:Inter1}--\eqref{eqn:Inter2} follow by iterating the equalities:
\begin{equation}
\label{eqn:ww3}
\beta_{\pm k} - \alpha_{\pm k} = \beta_{\pm (k-1)} - \alpha_{\pm (k-1)}
\end{equation}
\begin{equation}
\label{eqn:ww4}
\wgamma_{\pm k} - \walpha_{\pm k} = \wgamma_{\pm (k-1)} - \walpha_{\pm (k-1)} \pm \frac t{t_1t_2} \cdot \beta_{\pm (k-1)} \mp
\end{equation}
$$
\mp \left[\frac {r(m^2r - (r + 1)mt + t^2\be)}{2t_1t_2} + \frac {2mr - (r + 1)t}{2t_1t_2} (|\bu| - |\bu'|) + \frac {(|\bu| - |\bu'|)^2}{2t_1t_2}\right]
$$
In the above formulas, we recall that $\bd$ denotes the degree operator and $|\bu| = u_1+...+u_r$. We often replace these quantities by the following renormalizations: 
$$
\barbu = (\baru_1,...,\baru_r) \quad \text{where} \quad \baru_i = u_i + \frac {(r-1)t}2 \quad \Longrightarrow \quad B_0 \Big |_{H_\bu} = -|\barbu|
$$
$$
\widetilde{\bd} := \bd + \frac {r(r^2-1) t^2}{24t_1t_2} - \frac 1{2t_1t_2} \sum_{i=1}^r \bar{u}_i^2 \qquad \Longrightarrow \qquad \ L_0 \Big|_{H_\bu} = \widetilde{\bd}
$$
as in Subsection \ref{sub:explicitformulas}. To make the intersection theory part as clear as possible, we will focus on proving \eqref{eqn:ww1}--\eqref{eqn:ww2}, and then explain the differences that arise in formulas \eqref{eqn:ww3}--\eqref{eqn:ww4}. Write $\fZ_1  = \bigsqcup_{d = 0}^\infty \fZ_{d+1,d}$ and consider the spaces:
$$
\M_\bu \times \fZ_{1} \stackrel{\id \times \pi^\pm}\longrightarrow \M_{\bu} \times \M_{\bu'} \stackrel{\pi^\mp \times \id}\longleftarrow \fZ_{1} \times \M_{\bu'}
$$
By definition, the compositions $A_m \circ B_{\pm 1}$ and $B_{\pm 1} \circ A_m$ are given by the classes:
\begin{equation}
\label{eqn:a1}
 (-1)^{(r-1)\be} \left( \id \times \pi^\pm \right)_*\Big( [\fZ^\vir_1] \cdot c(\E_1, m) \Big)
\end{equation}
\begin{equation}
\label{eqn:a2}
(-1)^{r\e-\be} \ \left(\pi^\mp \times \id \right)_*\Big( [\fZ^\vir_1] \cdot c(\E_2, m) \Big)
\end{equation}
on $\M_\bu \times \M_{\bu'}$, respectively, where:
$$
\E_1 \Big |_{\F, \F_+ \subset \F_-} \ = \text{Ext}^1 \left(\F_\mp, \F(-\infty) \right)
$$
$$
\E_2 \Big |_{\F_+ \subset \F_-, \F'} = \text{Ext}^1 \left(\F', \F_\pm(-\infty) \right)
$$
If we apply \eqref{eqn:eformula}, we obtain the following equalities in $K$--theory:
$$
\E_1 = \left( \id \times \pi^\pm \right)^*(\E) \mp \left[ \sum_{i=1}^r \frac {e^{u_i-t}}{\L} - \left(1- \frac 1{e^{t_1}} \right)\left(1 - \frac 1{e^{t_2}} \right)\frac {\V}{\L} \right]
$$
$$
\E_2 \ = \ \left( \pi^\mp \times \id \right)^*(\E) \pm \left[ \sum_{i=1}^r \frac {\L}{e^{u'_i}} - \left(1- \frac 1{e^{t_1}} \right)\left(1 - \frac 1{e^{t_2}} \right)\frac {\L}{\V'} \right]
$$
where $\L$ is the tautological line bundle on $\fZ_1$. Therefore, in the notation \eqref{eqn:notation}:
\begin{equation}
\label{eqn:b1}
\alpha_{\pm 1} = (-1)^{(r-1)\be} \left( \id \times \pi^\pm \right)_* \left( [\fZ^\vir_1] \cdot \overline{\zeta(l - X - m)^{\mp 1}} \tau_{\bu}(l  + t - m)^{\mp 1}  \right) \qquad
\end{equation}
\begin{equation}
\label{eqn:b2}
\beta_{\pm 1} = (-1)^{r\e-\be} \left(\pi^\mp \times \id \right)_* \left( [\fZ^\vir_1] \cdot \overline{\zeta(X' - l - m)^{\pm 1}} \tau_{\bu'}(l+m)^{\pm 1} \right) \qquad
\end{equation}
To compute the push-forward \eqref{eqn:b1}, we invoke \eqref{eqn:imp} for $k=1$ and the choice of contours \eqref{eqn:normal3}. We assume that the $X$ variables (as well as the parameters $\bu$) are large, while the $X'$ variables (as well as the parameters $\bu'$) are small:
\begin{equation}
\label{eqn:c1}
\alpha_{\pm 1} = \mp\int_{X \gg |z| \gg X'} F_\pm(z)
\end{equation}
where:
$$
F_\pm(z) = \overline{\left[ \frac {\zeta(\pm z \mp X')}{\zeta(z - X - m)} \cdot \frac {\tau_{\bu'}(z+t\e)}{\tau_{\bu}(z + t - m)}\right]^{\pm 1}} \in H_\bu \otimes H_{\bu'}(z)
$$
Observe the following expansion, which is immediate from \eqref{eqn:zeta} and \eqref{eqn:tau}:
$$
F_\pm(z) = 1 \mp \frac {|\barbu'|-|\barbu| - r(m - t\be)}z + \frac 1{z^2} \Big( \mp t_1t_2 (\widetilde{\bd} - \widetilde{\bd'} ) + \frac {r(r\pm 1)m(m-t)}2 +
$$
\begin{equation}
\label{eqn:expansion}
\left.  + \frac {(r\pm 1)(2m-t)}2 |\barbu| - \frac {2rm - (r+1)t}2 |\barbu'| + \frac {(|\barbu| - |\barbu'|)^2}2\right) + O\left(\frac 1{z^3} \right) \qquad
\end{equation}
In similar fashion, one computes \eqref{eqn:b2} by using \eqref{eqn:imp} and the contours \eqref{eqn:normal4}:
\begin{equation}
\label{eqn:barbra}
\beta_{\pm 1} = \mp \int^{\wi}_{X \gg |z| \gg X'} \overline{ \left[ \frac {\zeta(X' - z - m)}{\zeta(\mp z \pm X)} \cdot \frac {\tau_{\bu'}(z + m)}{\tau_\bu \left(z + t\be \right)} \right]^{\pm 1}} 
\end{equation}
Using \eqref{eqn:identity} and the substitution $z \mapsto z - m + t\e$ (we must assume the parameter $m$ to be much smaller than the difference between the $z$ variables, or than $X,X',\bu,\bu'$ for that matter) the above formula yields:
\begin{equation}
\label{eqn:c2}
\beta_{\pm 1} = \mp \int^\wi_{X \gg |z| \gg X'} F_\pm(z)
\end{equation}
Comparing \eqref{eqn:c1} with \eqref{eqn:c2}, we observe that the integrands that compute $\alpha_{\pm 1}$ and $\beta_{\pm 1}$ are identical (among tautological classes, this is a feature which seems to be specific to the Ext bundle $\E$, and it holds for all Nakajima quiver varieties) and the only thing which differs between the two integrals is the residue at $\infty$: 
\begin{equation}
\label{eqn:unu}
\beta_{\pm 1} - \alpha_{\pm 1} = \mp \text{Res}_{z=\infty} \Big[ F_\pm(z) \Big] = |\barbu'|-|\barbu| - r(m - t\be)
\end{equation}
in virtue of \eqref{eqn:expansion}. This establishes \eqref{eqn:ww1}. As for \eqref{eqn:ww2}, the above analysis applies equally well, and we obtain the following formulas akin to \eqref{eqn:c1} and \eqref{eqn:c2}:
\begin{equation}
\label{eqn:enigma1}
\widetilde{\alpha}_{\pm 1} = \mp \int_{X \gg |z| \gg X'} \frac {z F_\pm(z)}{t_1t_2}
\end{equation}
\begin{equation}
\label{eqn:enigma2}
\widetilde{\beta}_{\pm 1} = \mp \int^\wi_{X \gg |z| \gg X'} \frac {(z - m + t\e) F_\pm(z)}{t_1t_2}
\end{equation}
The reason why \eqref{eqn:enigma1} differs from \eqref{eqn:enigma2} is the substitution $z \mapsto z - m + t\e$ that we applied to \eqref{eqn:barbra} in order to obtain \eqref{eqn:c2}. The whole reason why we introduced the class $\wgamma_{\pm 1} = \wbeta_{\pm 1} + \frac {m-t\e}{t_1t_2} \cdot \beta_{\pm 1}$ is that relation \eqref{eqn:enigma2} becomes:
\begin{equation}
\label{eqn:enigma3}
\wgamma_{\pm 1} = \mp \int^\wi_{X \gg |z| \gg X'} \frac {zF_\pm(z)}{t_1t_2}
\end{equation}
Subtracting relations \eqref{eqn:enigma1} and \eqref{eqn:enigma3} allows us to establish \eqref{eqn:ww2}:
\begin{equation}
\label{eqn:doi}
\wgamma_{\pm 1} - \walpha_{\pm 1} = \mp \text{Res}_{z=\infty} \left[ \frac {zF_\pm(z)}{t_1t_2} \right] = \text{RHS of \eqref{eqn:ww2}} 
\end{equation}
Along the same line of reasoning, Proposition \ref{prop:imp} implies that the compositions $A_m \circ B_{\pm k}$ and $B_{\pm k} \circ A_m$ are given by $c(\E,m)$ times the cohomology classes:
\begin{equation}
\label{eqn:cretu1}
\alpha_{\pm k} = (\mp 1)^k \int^\pm S(z_1,...,z_k) F_\pm(z_1)...F_\pm(z_k)
\end{equation}
\begin{equation}
\label{eqn:cretu2}
\beta_{\pm k} = (\mp 1)^k \int^{\mp, \wi} S(z_1,...,z_k) F_\pm(z_1)...F_\pm(z_k)
\end{equation}
respectively, where:
$$
S(z_1,...,z_k) = \frac {\prod_{1\leq i < j \leq k} \zeta(z_j-z_i)^{-1}}{(z_2-z_1+t)...(z_k-z_{k-1}+t)} 
$$
It is easy to see that $S$ has degree $-2$ in the variables $z_2,...,z_{k-1}$, while:
$$
S(z_1,...,z_k) = \frac 1{z_k} + O\left(\frac 1{z_k^2} \right) \qquad \qquad S(z_1,...,z_k) = - \frac 1{z_1} + O\left(\frac 1{z_1^2} \right)
$$
When we take the difference between \eqref{eqn:cretu1} and \eqref{eqn:cretu2}, we will pick up a sum of residues when $z_1,...,z_k$ pass around $\infty$. The residues in the variables $z_2,...,z_{k-1}$ vanish because of the property that the integrand has degree $\leq -2$ in these variables, while the other two residues contribute:
$$
\beta_{k} - \alpha_{k} = (-1)^k \int^\wi_{X \gg |z_k| \gg ... \gg |z_2| \gg X'} \text{Res}_{z_1=\infty}  \Big[S(z_1,...,z_k) F_\pm(z_1)...F_\pm(z_k) \Big] + 
$$
$$
+(-1)^k \int_{X \gg |z_{k-1}| \gg ... \gg |z_1| \gg X'} \text{Res}_{z_k=\infty}  \Big[S(z_1,...,z_k) F_\pm(z_1)...F_\pm(z_k) \Big] = \beta_{k-1} - \alpha_{k-1}
$$
and:
$$
\beta_{-k} - \alpha_{-k} =  \int^\wi_{X \gg |z_{1}| \gg ... \gg |z_{k-1}| \gg X'} \text{Res}_{z_k=\infty}  \Big[S(z_1,...,z_k) F_\pm(z_1)...F_\pm(z_k) \Big] + 
$$
$$
+ \int_{X \gg |z_{2}| \gg ... \gg |z_k| \gg X'} \text{Res}_{z_1=\infty}  \Big[S(z_1,...,z_k) F_\pm(z_1)...F_\pm(z_k) \Big] = \beta_{-k+1} - \alpha_{-k+1}
$$
In the above relations, we used \eqref{eqn:normal3} and \eqref{eqn:normal4} for the integrals $\int^\pm$. This proves relation \eqref{eqn:ww3}. As for \eqref{eqn:ww4}, the analogous analysis implies that the compositions $A_m \circ L_{\pm k}$ and $L_{\pm k} \circ A_m$ are given by $c(\E,m)$ times the classes:
\begin{equation}
\label{eqn:e-nigma1}
\widetilde{\alpha}_{\pm k} = (\mp 1)^k \int^\pm \widetilde{S}(z_1,...,z_k) F_\pm(z_1)...F_\pm(z_k)
\end{equation}
$$
\widetilde{\beta}_{\pm k} = (\mp 1)^k \int^{\mp,\wi} \left[ \widetilde{S}(z_1,...,z_k) - \frac {(m-t\e)S(z_1,...,z_k)}{t_1t_2} \right] F_\pm(z_1)...F_\pm(z_k)
$$
where:
$$
\widetilde{S}(z_1,...,z_k) = \frac {z_1+z_k-(k-1)rt}{2t_1t_2(z_2-z_1+t)...(z_k-z_{k-1}+t)} \prod_{1\leq i < j \leq k} \zeta(z_j-z_i)^{-1}
$$
Note that $\widetilde{S}$ has degree $\leq -2$ in the variables $z_2,...,z_{k-1}$, while:
\begin{equation}
\label{eqn:expansion1}
\widetilde{S}(z_1,...,z_k) = \frac 1{2t_1t_2} + \frac {z_1+z_{k-1} - (k-1)rt - t}{2t_1t_2 z_k} + O\left(\frac 1{z_k^2} \right)
\end{equation}
\begin{equation}
\label{eqn:expansion2}
\widetilde{S}(z_1,...,z_k) = - \frac 1{2t_1t_2} - \frac {z_2+z_{k} - (k-1)rt + t}{2t_1t_2 z_1} + O\left(\frac 1{z_1^2} \right)
\end{equation}
It makes sense to replace $\wbeta_{\pm k}$ by $\wgamma_{\pm k} = \wbeta_{\pm k} + \frac {m-t\e}{t_1t_2} \cdot \beta_{\pm k} $, for which:
\begin{equation}
\label{eqn:e-nigma2}
\wgamma_{\pm k} = (\mp 1)^k \int^{\mp,\wi} \widetilde{S}(z_1,...,z_k) F_\pm(z_1)...F_\pm(z_k)
\end{equation}
Then for all $k>1$, we have:
$$
\wgamma_{k} - \walpha_{k} = (-1)^k \int^\wi_{X \gg |z_k| \gg ... \gg |z_2| \gg X'} \text{Res}_{z_1=\infty}  \Big[\widetilde{S}(z_1,...,z_k) F_\pm(z_1)...F_\pm(z_k) \Big] + 
$$
$$
+(-1)^k \int_{X \gg |z_{k-1}| \gg ... \gg |z_1| \gg X'} \text{Res}_{z_k=\infty}  \Big[\widetilde{S}(z_1,...,z_k) F_\pm(z_1)...F_\pm(z_k) \Big] = 
$$
$$
= \left(\wgamma_{k-1} + \frac {|\barbu| - |\barbu'| + r(m-t) + t}{2t_1t_2} \beta_{k-1} \right) - \left(\walpha_{k-1} + \frac {|\barbu| - |\barbu'| + r(m-t) - t}{2t_1t_2} \alpha_{k-1} \right)
$$
where the residue counts follow by \eqref{eqn:expansion} and \eqref{eqn:expansion1}--\eqref{eqn:expansion2}. In similar fashion:
$$
\wgamma_{-k} - \walpha_{-k} =  \int^\wi_{X \gg |z_{1}| \gg ... \gg |z_{k-1}| \gg X'} \text{Res}_{z_k=\infty}  \Big[\widetilde{S}(z_1,...,z_k) F_\pm(z_1)...F_\pm(z_k) \Big] + 
$$
$$
+ \int_{X \gg |z_{2}| \gg ... \gg |z_k| \gg X'} \text{Res}_{z_1=\infty}  \Big[\widetilde{S}(z_1,...,z_k) F_\pm(z_1)...F_\pm(z_k) \Big] = 
$$
$$
= \left(\wgamma_{-k+1} + \frac {|\barbu'| - |\barbu| - rm - t}{2t_1t_2} \beta_{-k+1} \right) - \left(\walpha_{-k+1} + \frac {|\barbu'| - |\barbu| - rm + t}{2t_1t_2} \alpha_{-k+1} \right)
$$
We conclude that:
$$
\wgamma_{\pm k} - \walpha_{\pm k} = \wgamma_{\pm (k-1)} - \walpha_{\pm (k-1)} \pm 
$$
$$
\pm \frac {|\bu| - |\bu'| + r(m-t\e) - t}{2t_1t_2} \Big(\beta_{\pm (k-1)} - \alpha_{\pm (k-1)}\Big) \pm \frac t{t_1t_2} \beta_{\pm (k-1)}
$$
Plugging in \eqref{eqn:ww1} and \eqref{eqn:ww3} to evaluate the first term on the second line gives us precisely \eqref{eqn:ww4}, thus completing the proof. \\
\end{proof}

\section{Traces of intertwiners and the partition function}\label{sec:reptheory}

\subsection{}\label{sub:corollary}

In order to prove Corollary \ref{cor:main}, we must establish how the generators of the Heisenberg-Virasoro algebra commute with the group-like elements $g_\pm = g_\pm(1)$. The following formulas are well-known and straightforward exercises:
\begin{equation}
\label{eqn:heiscomm}
\left[ B_{\pm k}, g_-^a \right] = - \e a r t_1 t_2 \cdot  g_-^a \qquad \qquad \left[ B_{\pm k}, g_+^a \right] = \be a r t_1 t_2  \cdot g_+^a
\end{equation}
for all $k>0$ and all constants $a \in \BF$. In the above formulas, we used the fact that $c_1 = - r t_1t_2$ in all our representations, as follows from \eqref{eqn:central2}. The following formulas are a bit more involved, but also straightforward:
\begin{equation}
\label{eqn:vircomm1}
\left[L_{\pm k} - L_{\pm (k-1)}, g_-^a \right] \ = \ \left( \pm  a B_{\pm k - \e} + \frac {\e a^2 \delta_{k>1} rt_1t_2 }2 \right) g_-^a
\end{equation}
\begin{equation}
\label{eqn:vircomm2}
\left[ L_{\pm k} - L_{\pm (k-1)}, g_+^a \right] \ = \ \left( \pm  a B_{\pm k + \be} + \frac {\be a^2 \delta_{k>1} rt_1t_2 }2 \right) g_+^a
\end{equation}
To prove \eqref{eqn:vircomm1}, note that:
$$
X = a \sum_{k=1}^\infty \frac {B_{-k}}k \quad \Longrightarrow \quad [L_{\pm k} - L_{\pm (k-1)}, X] \ = \ \pm a B_{\pm k - \e} \ =: \ Y
$$
where we used \eqref{eqn:Heisvir} to compute the commutator. Using \eqref{eqn:Heis}, we have:
$$
[X , Y] \ = \ - \e a^2 \delta_{k>1} c_1 = \e a^2 \delta_{k>1} rt_1t_2 \ =: \ Z
$$
and note that $Z$ is central. Therefore, we have:
$$
\left[L_{\pm k} - L_{\pm (k-1)}, g_-^a \right] = \left[L_{\pm k} - L_{\pm (k-1)}, \exp(X) \right] = \sum_{n=0}^\infty \left[L_{\pm k} - L_{\pm (k-1)}, \frac {X^n}{n!} \right] =
$$
$$
= \sum_{p,q = 0}^\infty \frac {X^p Y X^q}{(p+q+1)!} = \sum_{p,q = 0}^\infty \frac { Y X^{p+q}}{(p+q+1)!} + \sum_{p,q,r=0}^\infty \frac {Z X^{p+q+r}}{(p+q+r+2)!} = 
$$
$$
= \sum_{n= 0}^\infty \frac {(n+1) Y X^n}{(n+1)!} + \sum_{n=0}^\infty \frac {(n+1)(n+2)Z X^{n}}{2(n+2)!} = \left(Y + \frac {Z}2 \right) \exp(X) = \left(Y + \frac {Z}2 \right)  g_-^a
$$
which proves \eqref{eqn:vircomm1}. Relation \eqref{eqn:vircomm2} is proved analogously. \\

\begin{proof}{\bf of Corollary \ref{cor:main}:} Set $A_m = A_m(1)$. Since $g_\pm$ is invertible, we may write:
$$
A_m = g_-^{\frac m{t_1t_2}} \cdot \Omega_m \cdot g_+^{\frac {t-m}{t_1t_2}}
$$
for some operator $\Omega_m : M_\bu \cong H_\bu \rightarrow H_{\bu'} \cong M_{\bu'}$. With the following convention:
$$
\Omega_m(x) = \sum_{d,d' = 0}^\infty \Omega_m|^{d'}_d \cdot x^{d - d' + \lev_{\bu} - \lev_{\bu'}}
$$
the commutation relations $[B_0,\Omega_m(x)]$ and $[L_0, \Omega_m(x)]$ predicted by \eqref{eqn:inter1}--\eqref{eqn:inter2} follow automatically from our conventions \eqref{eqn:cartan2} and \eqref{eqn:central2}:
$$
B_0 \Big |_{H_\bu \cong M_\bu} = - |\barbu|
$$
$$
L_0 \Big |_{H_{\bu} \cong M_{\bu}} = \bd + \lev_\bu
$$
Hence the general case of \eqref{eqn:inter1}--\eqref{eqn:inter2} will be proved once we establish the relations:
\begin{equation}
\label{eqn:Inter5}
[B_{\pm k} - B_0, \Omega_m] = 0
\end{equation}
\begin{equation}
\label{eqn:Inter6}
[L_{\pm k} - L_{\pm (k-1)}, \Omega_m] = 
\end{equation}
$$
= \mp \left[ \frac {r(r-1)m(m-t)}{2t_1t_2} \Omega_m + \frac {(r-1)(2m-t)}{2t_1t_2} [\Omega_m, B_0] + \frac {[[\Omega_m,B_0],B_0]}{2t_1t_2} \right]
$$
To prove formula \eqref{eqn:Inter5}, recall relation \eqref{eqn:Inter1}, which states that:
$$
- r(m-t\be)  g_-^{\frac m{t_1t_2}} \cdot \Omega_m \cdot g_+^{\frac {t-m}{t_1t_2}} = \left[\B, g_-^{\frac m{t_1t_2}} \cdot \Omega_m \cdot g_+^{\frac {t-m}{t_1t_2}} \right] = 
$$
$$
= \left[ \B, g_-^{\frac m{t_1t_2}} \right] \Omega_m \cdot g_+^{\frac {t-m}{t_1t_2}} + g_-^{\frac m{t_1t_2}} \Big[\B, \Omega_m \Big] g_+^{\frac {t-m}{t_1t_2}} + g_-^{\frac m{t_1t_2}} \Omega_m \left[\B, g_+^{\frac {t-m}{t_1t_2}} \right]
$$
where $\B = B_{\pm k} - B_0$. The middle term on the second line is the one we wish to compute. We can use \eqref{eqn:heiscomm} to rewrite the right hand side as:
$$
= - \e r m g_-^{\frac m{t_1t_2}} \cdot \Omega_m \cdot g_+^{\frac {t-m}{t_1t_2}} + g_-^{\frac m{t_1t_2}} \Big[\B, \Omega_m \Big] g_+^{\frac {t-m}{t_1t_2}} - \be r (m-t) g_-^{\frac m{t_1t_2}} \cdot \Omega_m \cdot g_+^{\frac {t-m}{t_1t_2}} 
$$
We obtain $[\B, \Omega_m]=0$, which is precisely the content of \eqref{eqn:Inter5}, as required. Let us now prove \eqref{eqn:Inter6}. We start with the case $k=1$, so let us write relation \eqref{eqn:Inter2} as:
$$
\mp \left[\frac {r(r\pm 1)m(m-t)}{2t_1t_2} A_m - \frac {(r\pm 1)(2m-t)}{2t_1t_2} B_0A_m + \frac {2rm - (r+1)t}{2 t_1 t_2} A_m B_0 \right] \mp
$$
\begin{equation}
\label{eqn:henry}
\mp  \frac {[[A_m, B_0], B_0]}{2t_1t_2} - \frac {m-t\e}{t_1t_2} B_{\pm 1}A_m  = \left [\L, A_m \right] = \left [\L,  g_-^{\frac m{t_1t_2}} \Omega_m g_+^{\frac {t-m}{t_1t_2}}  \right] = 
\end{equation}
$$
= \left [\L,  g_-^{\frac m{t_1t_2}} \right] \Omega_m g_+^{\frac {t-m}{t_1t_2}} + g_-^{\frac m{t_1t_2}} \left [\L, \Omega_m \right] g_+^{\frac {t-m}{t_1t_2}} + g_-^{\frac m{t_1t_2}} \Omega_m \left [\L,  g_+^{\frac {t-m}{t_1t_2}}  \right]
$$
where $\L = L_{\pm 1} - L_0$. As before, our goal is to compute the middle term on the last line. We may use \eqref{eqn:vircomm1} and \eqref{eqn:vircomm2} to rewrite the right hand side of \eqref{eqn:henry} as:
$$
m \frac {\pm B_{\pm 1 - \e}}{t_1t_2} \cdot g_-^{\frac m{t_1t_2}} \Omega_m g_+^{\frac {t-m}{t_1t_2}} + g_-^{\frac m{t_1t_2}} \left [\L, \Omega_m \right] g_+^{\frac {t-m}{t_1t_2}} + 
$$
\begin{equation}
\label{eqn:tudor}
+ g_-^{\frac m{t_1t_2}} \Omega_m \cdot (m-t) \frac {\mp B_{\pm 1 + \be}}{t_1t_2} \cdot g_+^{\frac {t-m}{t_1t_2}} 
\end{equation}
In the last term of \eqref{eqn:tudor}, we can move the fraction all the way to the left of $g_-^{\frac m{t_1t_2}}$ using \eqref{eqn:inter1}. Therefore, we obtain:
$$
\text{LHS of \eqref{eqn:henry}} = m \frac {\pm B_{\pm 1 - \e}}{t_1t_2} \cdot g_-^{\frac m{t_1t_2}}\Omega_m g_+^{\frac {t-m}{t_1t_2}} +  g_-^{\frac m{t_1t_2}} \left [\L, \Omega_m \right] g_+^{\frac {t-m}{t_1t_2}} + 
$$
$$
+ (m-t) \left[ \frac {\mp B_{\pm 1 + \be}}{t_1t_2} g_-^{\frac m{t_1t_2}}\Omega_m g_+^{\frac {t-m}{t_1t_2}} - \frac {\e mr}{t_1t_2}g_-^{\frac m{t_1t_2}}\Omega_m g_+^{\frac {t-m}{t_1t_2}} \mp \frac {g_-^{\frac m{t_1t_2}} [\Omega_m,B_0] g_+^{\frac {t-m}{t_1t_2}}}{t_1t_2} \right]
$$
Isolating the last term on the first line, we conclude that $g_-^{\frac m{t_1t_2}} \left [\L, \Omega_m \right] g_+^{\frac {t-m}{t_1t_2}}$ equals:
$$
\mp g_-^{\frac m{t_1t_2}} \left[ \frac {r(r-1)m(m-t)}{2t_1t_2} \Omega_m + \frac {(2m - t)(r-1)}{2t_1t_2} [\Omega_m,B_0] + \frac {[[\Omega_m,B_0],B_0]}{2t_1t_2} \right] g_+^{\frac {t-m}{t_1t_2}}
$$
which is precisely \eqref{eqn:Inter6} for $k=1$. In order to prove \eqref{eqn:Inter6} for $k>1$, use \eqref{eqn:Inter2}:
$$
\mp  \left[ \frac {r(rm^2 - (r + 1)mt + t^2\be)}{2t_1t_2} + \frac {2mr - (r + 1)t}{2t_1t_2} [A_m,B_0] + \frac {[[A_m,B_0],B_0]}{2t_1t_2} \right] + 
$$
\begin{equation}
\label{eqn:macbeth}
+  \left(\frac {m - t\be}{t_1t_2}B_{\pm (k-1)}  - \frac {m-t\e}{t_1t_2} B_{\pm k}\right)A_m =  \left [\L',  g_-^{\frac m{t_1t_2}} \Omega_m g_+^{\frac {t-m}{t_1t_2}}  \right] = 
\end{equation}
$$
=  \left [\L',  g_-^{\frac m{t_1t_2}} \right] \Omega_m g_+^{\frac {t-m}{t_1t_2}} + g_-^{\frac m{t_1t_2}} \left [\L', \Omega_m \right] g_+^{\frac {t-m}{t_1t_2}} + g_-^{\frac m{t_1t_2}} \Omega_m \left [\L',  g_+^{\frac {t-m}{t_1t_2}}  \right]
$$
where $\L' = L_{\pm k} - L_{\pm (k-1)}$. As before, our goal is to evaluate the middle term of the right hand side, which we achieve by applying \eqref{eqn:vircomm1}--\eqref{eqn:vircomm2} to the other two terms. The right hand side of expression \eqref{eqn:macbeth} equals:
$$
= m \frac {\pm 2 B_{\pm k - \e} + \e r m}{2t_1t_2} \cdot g_-^{\frac m{t_1t_2}} \Omega_m g_+^{\frac {t-m}{t_1t_2}} + g_-^{\frac m{t_1t_2}} \left [\L', \Omega_m \right] g_+^{\frac {t-m}{t_1t_2}} + 
$$
$$
+ g_-^{\frac m{t_1t_2}} \Omega_m \cdot (m-t) \frac {\mp 2 B_{\pm k + \be} + r \be (m-t)}{2t_1 t_2} \cdot g_+^{\frac {t-m}{t_1t_2}}
$$
In the second line, we can move the fraction all the way to the left of $g_-^{\frac m{t_1t_2}}$ using \eqref{eqn:inter1}. In doing so, we obtain:
$$
\text{LHS of \eqref{eqn:macbeth}} = m \frac {\pm 2 B_{\pm k - \e} + \e r m}{2t_1t_2} g_-^{\frac m{t_1t_2}} \Omega_m g_+^{\frac {t-m}{t_1t_2}} + g_-^{\frac m{t_1t_2}} \left [\L', \Omega_m \right] g_+^{\frac {t-m}{t_1t_2}} +
$$
$$
+ (m-t) \frac {\mp 2 B_{\pm k + \be} + \be r(m-t)}{2t_1 t_2}  g_-^{\frac m{t_1t_2}} \cdot \Omega_m \cdot g_+^{\frac {t-m}{t_1t_2}} -
$$
$$
-  \frac {\e r m(m-t)}{t_1t_2} g_-^{\frac m{t_1t_2}}  \cdot \Omega_m \cdot g_+^{\frac {t-m}{t_1t_2}} \mp \frac {m-t}{t_1t_2}  g_-^{\frac m{t_1t_2}} \cdot [\Omega_m,B_0] \cdot g_+^{\frac {t-m}{t_1t_2}}
$$
Isolating the last term on the first line, we see that $g_-^{\frac m{t_1t_2}} \left [\L', \Omega_m \right] g_+^{\frac {t-m}{t_1t_2}}  = $
$$
\mp g_-^{\frac m{t_1t_2}} \left[ \frac {r(rm^2 - (r + 1)mt + t^2\be)}{2t_1t_2} \Omega_m + \frac {2mr - (r + 1)t}{2t_1t_2} [\Omega_m,B_0] + \frac {[[\Omega_m,B_0],B_0]}{2t_1t_2} \right. +
$$ 
$$
+ \left. \frac {\e rm^2}{2t_1t_2} \Omega_m + \frac {\be r (m-t)^2}{2t_1t_2} \Omega_m - \frac {\e r m(m-t)}{t_1t_2} \Omega_m - [\Omega_m,B_0] \frac {m-t}{t_1t_2} \right] g_+^{\frac {t-m}{t_1t_2}}
$$
The right hand side is precisely \eqref{eqn:Inter6}, and with this, the proof is complete. \\
\end{proof}

\subsection{} We will now set up the proof of Corollary \ref{cor:agt}, so we assume $r=2$ for the remainder of this Section (this assumption will be used in Proposition \ref{prop:commute}). To do so, let us observe the following easy-to-prove formulas for commuting creation and annihilation vertex operators past each other:
$$
g_+^a(x) g^b_-(y) g^{-a}_+(x) g^{-b}_-(y) = \exp \left(ab \sum_{k=1}^\infty \frac {[B_k, B_{-k}]}{k^2} \cdot \frac {y^k}{x^k} \right) \Longrightarrow 
$$
\begin{equation}
\label{eqn:enigmatic}
\Longrightarrow g^{a}_+(x) g^b_-(y) = g^b_-(y) g_+^a(x) \left(1 - \frac yx \right)^{abrt_1t_2}
\end{equation}
A slightly more involved exercise is the following computation: \\

\begin{proposition}
\label{prop:commute}

With $\beta$ as in \eqref{eqn:beta}, we have:
\begin{equation}
\label{eqn:commute1}
g^a_+(z) \Omega_m(x) \ = \ \left(1 - \frac xz\right)^{a\beta} \Omega_m(x) g^a_+(z)
\end{equation}
\begin{equation}
\label{eqn:commute2}
\Omega_m(x) g^a_-(z) = \left(1 - \frac zx\right)^{-a\beta} g^a_-(z) \Omega_m(x) 
\end{equation}
We will always have $z \neq x$, so the right hand sides make sense. \\
\end{proposition}

\begin{proof} It is clear that the Liouville vertex operator is uniquely determined by properties \eqref{eqn:inter1}--\eqref{eqn:inter2}. The reason for this is that the Verma module is generated by $\{L_{-k},B_{-k}\}_{k\in \BN}$ acting on the vacuum vector $\vac$. Therefore, any vertex operator satisfying \eqref{eqn:inter1}--\eqref{eqn:inter2} is completely determined by the vector $v = \Omega_m(x)\cdot \vac$. However, this vector can be uniquely reconstructed from its vacuum coefficient, since \eqref{eqn:inter1}--\eqref{eqn:inter2} imply that for all $k>0$ we have:
$$
B_k \cdot v = \beta x^k \cdot v, \qquad L_k\cdot v = \left(x^{k+1}\frac {\partial}{\partial x} - \lambda k x^k\right) v
$$
Therefore, to prove \eqref{eqn:commute1}, it is enough to prove that:
$$
\wOmega_m(x) =  g_+(z)^a \cdot \frac {\Omega_m(x)}{\left(1 - \frac xz\right)^{a\beta}} \cdot g_+(z)^{-a}
$$
satisfies properties \eqref{eqn:inter1}--\eqref{eqn:inter2}. To do so, we will use \eqref{eqn:inter1} and \eqref{eqn:heiscomm} to compute:
$$
\left[B_{k}, \wOmega_m(x) \right] = \left[B_{k}, g^a_+(z) \right] \frac {\Omega_m(x)}{\left(1 - \frac xz\right)^{a\beta}} g^{-a}_+(z) + g^a_+(z) \left[B_{k}, \frac {\Omega_m(x)}{\left(1 - \frac xz\right)^{a\beta}} \right]  g^{-a}_+(z) +
$$
$$
+ \frac {\Omega_m(x)}{\left(1 - \frac xz\right)^{a\beta}} \left[B_{k}, g^{-a}_+(z) \right] = \beta g^a_+(z) \frac {\Omega_m(x)}{\left(1 - \frac xz\right)^{a\beta}} g^{-a}_+(z) = \beta \wOmega_m(z)
$$
Meanwhile, by iterating \eqref{eqn:vircomm1} one can prove that:
$$
[L_k, g^a_+(z)] = \left(\frac {a^2rt_1t_2 \max(-1-k,0)}2  z^{k}  - a \sum_{i=1}^\infty  B_{k+i} z^{-i}\right) g^a_+(z)
$$
With the above relation and \eqref{eqn:inter2}, we can write $\left[L_k, \wOmega_m(x) \right]$ as:
$$
\left[L_k, g^a_+(z) \right] \frac {\Omega_m(x)}{\left(1 - \frac xz\right)^{a\beta}} g^{-a}_+(z) + g^a_+(z) \left[L_k, \frac {\Omega_m(x)}{\left(1 - \frac xz\right)^{a\beta}} \right]  g^{-a}_+(z) + g^a_+(z) \frac {\Omega_m(x)}{\left(1 - \frac xz\right)^{a\beta}} \left[L_k, g^{-a}_+(z) \right] 
$$
$$
= \left(\frac {a^2rt_1t_2 \max(-1-k,0)}2  z^{k}  - a \sum_{i=1}^\infty  B_{k+i} z^{-i}\right) g^a_+(z) \frac {\Omega_m(x)}{\left(1 - \frac xz\right)^{a\beta}} g^{-a}_+(z) +
$$
$$
+ g^a_+(z) \left(x^{k+1}\frac {\partial}{\partial x} - \lambda k x^k + \frac {a\beta x^{k} \cdot \frac xz}{1 - \frac xz} \right) \frac {\Omega_m(x)}{\left(1 - \frac xz\right)^{a\beta}} g^{-a}_+(z)+
$$
$$
+ g^a_+(z) \frac {\Omega_m(x)}{\left(1 - \frac xz\right)^{a\beta}} \left(\frac {a^2rt_1t_2 \max(-1-k,0)}2  z^{k}  - a \sum_{i=1}^\infty  B_{k+i} z^{-i}\right) g^{-a}_+(z) = 
$$
$$
= \left(x^{k+1}\frac {\partial}{\partial x} - \lambda k x^k + \frac {a\beta x^{k} \cdot \frac xz}{1 - \frac xz} + a^2rt_1t_2 \max(-1-k,0)  z^{k}\right) g^a_+(z) \frac {\Omega_m(x)}{\left(1 - \frac xz\right)^{a\beta}} g^{-a}_+(z)+
$$
$$
- a \beta \sum_{i=1}^{\infty} x^{k+i} z^{-i} g^a_+(z) \frac {\Omega_m(x)}{\left(1 - \frac xz\right)^{a\beta}} g^{-a}_+(z) - \sum_{i=-1}^{k+1}a^2rt_1t_2 z^{k} g^a_+(z)\frac {\Omega_m(x)}{\left(1 - \frac xz\right)^{a\beta}} g^{-a}_+(z) 
$$
The above equals $\left(x^{k+1}\frac {\partial}{\partial x} - \lambda k x^k \right)\Omega_m(x)$, 
and hence $\wOmega_m(x) = \Omega_m(x)$. This proves \eqref{eqn:commute1}. Relation \eqref{eqn:commute2} is proved analogously, so we leave it to the interested reader. \\
\end{proof}

\begin{proof}{\bf of Corollary \ref{cor:agt}:} By definition, the Nekrasov partition function equals:
$$
Z = \Tr \left(Q^\bd A_{m_1}(x_1) ... A_{m_k}(x_k) \right) =
$$
$$
= \Tr \left(Q^\bd \cdot g_-^{\frac {m_1}{t_1t_2}}(x_1) \Omega_{m_1}(x_1) g_+^{\frac {t-m_1}{t_1t_2}}(x_1) \cdot ... \cdot g_-^{\frac {m_k}{t_1t_2}}(x_k) \Omega_{m_k}(x_k) g_+^{\frac {t-{m_k}}{t_1t_2}}(x_k) \right) 
$$
where in the last equality we used \eqref{eqn:factor}. We may now use \eqref{eqn:enigmatic}, \eqref{eqn:commute1} and \eqref{eqn:commute2} to move all the creation vertex operators to the left and all the annihilation vertex operators to the right (recall that $r=2$):
\begin{equation}
\label{eqn:nicky}
Z = \prod_{1\leq i < j \leq k} \left(1-\frac {x_j}{x_i} \right)^{\frac {-2(m_i-t)m_j - \beta(m_i-t) - \beta m_j}{t_1t_2}} \cdot Z'(x_1,...,x_k) \qquad
\end{equation}
where:
$$
Z'(x_1,...,x_k) = \Tr \left(Q^\bd \cdot g_-^{\frac {m_1}{t_1t_2}}(x_1)... g_-^{\frac {m_k}{t_1t_2}}(x_k) \cdot \Omega \cdot g_+^{\frac {t-{m_1}}{t_1t_2}}(x_1) ... g_+^{\frac {t-{m_k}}{t_1t_2}}(x_k) \right)
$$
where $\Omega = \Omega_{m_1}(x_1) ... \Omega_{m_k}(x_k)$. We may now use the relations:
$$
Q^\bd g_\pm^a(x) = g_\pm^a(Qx) Q^\bd 
$$
to write:
$$
Z'(x_1,...,x_k) = \Tr \left(g_-^{\frac {m_1}{t_1t_2}}(Qx_1)... g_-^{\frac {m_k}{t_1t_2}}(Qx_k) \cdot Q^\bd \cdot \Omega \cdot g_+^{\frac {t-{m_1}}{t_1t_2}}(x_1) ... g_+^{\frac {t-{m_k}}{t_1t_2}}(x_k) \right)
$$
The fundamental property of the trace implies that we can move factors from the front to the back of the trace without changing its value:
$$
Z'(x_1,...,x_k) = \Tr \left(Q^\bd \cdot \Omega \cdot g_+^{\frac {t-{m_1}}{t_1t_2}}(x_1) ... g_+^{\frac {t-{m_k}}{t_1t_2}}(x_k) \cdot g_-^{\frac {m_1}{t_1t_2}}(Qx_1)... g_-^{\frac {m_k}{t_1t_2}}(Qx_k) \right)
$$
Using \eqref{eqn:enigmatic} and \eqref{eqn:commute2}, the above formula can be written as:
$$
Z'(x_1,...,x_k) = \prod_{1\leq i,j \leq k} \left(1- \frac {Qx_j}{x_i} \right)^{\frac {-2(m_i-t)m_j - \beta m_j}{t_1t_2}} \cdot
$$
$$
\Tr \left(Q^\bd \cdot  g_-^{\frac {m_1}{t_1t_2}}(Qx_1)... g_-^{\frac {m_k}{t_1t_2}}(Qx_k) \cdot \Omega \cdot g_+^{\frac {t-{m_1}}{t_1t_2}}(x_1) ... g_+^{\frac {t-{m_k}}{t_1t_2}}(x_k) \right)
$$
Iterating this argument infinitely many times allows us to write:
$$
Z'(x_1,...,x_k) = \frac {\Tr \left(Q^\bd \cdot \Omega \cdot g_+^{\frac {t-{m_1}}{t_1t_2}}(x_1) ... g_+^{\frac {t-{m_k}}{t_1t_2}}(x_k) \right)}{\prod_{1\leq i,j \leq k} \left( \frac {Qx_j}{x_i};Q \right)_\infty^{\frac {2(m_i-t)m_j + \beta m_j}{t_1t_2}}}
$$
Running the similar argument with the product of $g_+$'s (that is, moving the product to the very left of the trace, then commuting it past $Q^\bd$ and $\Omega$) gives us:
$$
Z'(x_1,...,x_k) = \frac {\Tr (Q^\bd \cdot \Omega)}{\prod_{1\leq i,j \leq k} \left( \frac {Qx_j}{x_i};Q \right)_\infty^{\frac {2(m_i-t)m_j + \beta m_j + \beta(m_i-t)}{t_1t_2}}}
$$
Together with \eqref{eqn:nicky}, this implies \eqref{eqn:agt}. 
\end{proof}

\section{Appendix}\label{sec:appendix}

\begin{proof}{\bf of Proposition \ref{prop:belong}:} One must first prove that $C_m \in \S$, i.e. that it can be written in the form \eqref{eqn:defshuf} for a numerator which satisfies the wheel conditions. This is elementary and closely follows the corresponding argument in Proposition 6.2 of \cite{shuf}, so we leave it as an exercise to the interested reader. Let us prove the stronger statement that $C_m \in \S_\sma$. Call a polynomial $m \in \BF[z_1,...,z_k]$ {\bf good} if $C_m \in \S_\sma$, and write:
$$
I \subset \BF[z_1,...,z_k]
$$
for the vector space of good polynomials. We will prove that $I = \BF[z_1,...,z_k]$ by induction on $k$. The induction hypothesis implies that any multiple of $z_{i+1}-z_{i} +t$ lies in $I$, for all $1\leq i < k$, since $\S_\sma$ is a subalgebra. Thus we conclude that:
$$
\Big( z_2 -  z_1 + t,...,z_{k}-z_{k-1}+t \Big) \subset I 
$$
Therefore, in order to complete the induction step, it is enough to prove that for any polynomial in one variable $f(z)$, there exists a good polynomial $m$ such that $m(z-t,...,z - kt) = f(z)$. To this end, note that:
\begin{equation}
\label{eqn:ay}
\sym \left [ M(z_1,...,z_k)\prod_{1 \leq i <  j \leq k} \zeta(z_j-z_i) \right] 
\end{equation}
lies in $\S_\sma$ for any polynomial $M$. Observe that we have changed the order of the variables in the above expression, as opposed from \eqref{eqn:ideal}. In particular, take: 
$$
M = f(z_1) \prod_{1\leq i < j \leq k} (z_i-z_j+t_1) \prod_{1\leq i < j  \leq k} (z_i-z_j+t_2)  \prod_{1\leq i < j - 1 \leq k-1} (z_j-z_i+t) 
$$
for any $f \in \BF[z]$. We observe that the shuffle element \eqref{eqn:ay} equals:
$$
\sym \left[ \frac {f(z_1)\prod_{1\leq i < j \leq k} (z_j-z_i+t_1)(z_j-z_i+t_2) (z_j-z_i-t)}{\prod_{i=1}^{k-1} (z_{i+1}-z_{i}+t)} \prod_{1\leq i < j \leq k} \zeta(z_i-z_j) \right] 
$$
and conclude that:
$$
m(z_1,...,z_k) := f(z_1)\prod_{1\leq i < j \leq k} (z_j-z_i+t_1)(z_j-z_i+t_2) (z_j-z_i-t) 
$$
is good. Since the polynomial $m(z - t,...,z - kt) = \left( f(z)\cdot \text{constant} \right)$ can be made equal to any polynomial in $z$, the induction step is complete. \\
\end{proof}

\begin{proof}{\bf of Lemma \ref{lem:improve}:} Recall the definition of the shuffle product in \eqref{eqn:shufprod}:
\begin{equation}
\label{eqn:comm}
[R,R'] = \frac 1{k!\cdot k'!} \ \sym \Big[ P(z_1,...,z_{k+k'}) - P'(z_1,...,z_{k+k'}) \Big]
\end{equation}
where:
\begin{equation}
\label{eqn:emma}
P = R(z_1,...,z_k) R'(z_{k+1},...,z_{k+k'})  \prod_{1 \leq i \leq k < j \leq k+k'} \zeta(z_i-z_j) 
\end{equation}
\begin{equation}
\label{eqn:stone}
P' = R(z_1,...,z_k) R'(z_{k+1},...,z_{k+k'})  \prod_{1 \leq i \leq k < j \leq k+k'} \zeta(z_j-z_i) 
\end{equation}
When computing $l\ideg_{[R,R']}$, we need to compute the $y$--degree of the specialization: 
\begin{equation}
\label{eqn:evaleval}
\Big( P-P' \Big) \Big |_{z_{i_1} \mapsto y - t_1,..., z_{i_l} \mapsto y - l t_1}
\end{equation}
for any set $S = \{i_1, ..., i_l\} \subset \{1,...,k+k'\}$. The bounds on the second and third lines of \eqref{eqn:degreeimprove} come into play when $S \supset \{1,...,k\}$ or $S \supset \{k+1,...,k+k'\}$. Without loss of generality, we will consider only the first of these two situations. We have:
$$
P \Big |_{z_{1} \mapsto y - t_1, ..., z_k \mapsto y - kt_1, z_{k+1} \mapsto y - (k+1)t_1,...,z_{l} \mapsto y - l t_1} = R(y - t_1,..., y - k t_1) \cdot
$$
$$
R'(y - (k+1)t_1,...,y-l t_1, z_{l+1},...,z_{k+k'}) \prod_{k < j \leq l}^{1\leq i \leq k} \zeta(j t_1 - i t_1) \prod_{l < j \leq k+k'}^{1\leq i \leq k} \zeta(y - i t_1 - z_j) =
$$
$$
R(y - t_1,..., y - k t_1)R'(y - (k+1)t_1,...,y-l t_1, z_{l+1},...,z_{k+k'}) \prod_{k < j \leq l}^{1\leq i \leq k} \zeta(j t_1 - i t_1) + 
$$
\begin{equation}
\label{eqn:bradley}
+ O \left( y^{k\ideg_R + (l-k)\ideg_{R'}-1} \right)
\end{equation}
where we used \eqref{eqn:limit}. The reason why we only consider the evaluation $z_{i} \mapsto y-it_1$ instead of the more general \eqref{eqn:evaleval} is because $\zeta(-t_1) = 0$. Similarly, we have:
$$
P' \Big |_{z_{1} \mapsto y - (l-k+1)t_1,..., z_k \mapsto y - lt_1, z_{k+1} \mapsto y - t_1, ..., z_{l} \mapsto y - (l-k) t_1} = R(y - (l-k+1)t_1,..., y - l t_1) \cdot
$$
$$
R'(y -  t_1,...,y - (l-k) t_1, z_{l+1},...,z_{k+k'}) \prod_{k < j \leq l}^{1\leq i \leq k} \zeta(jt_1 - it_1) \prod_{l < j \leq k+k'}^{l - k < i \leq l} \zeta(z_j - y + i t_1) =
$$
$$
R(y - (l-k+1)t_1,..., y - l t_1)R'(y -  t_1,...,y - (l-k) t_1, z_{l+1},...,z_{k+k'}) \prod_{k < j \leq l}^{1\leq i \leq k} \zeta(j t_1 - i t_1) + 
$$
\begin{equation}
\label{eqn:cooper}
+ O \left( y^{k\ideg_R + (l-k)\ideg_{R'}-1} \right)
\end{equation}
The difference between \eqref{eqn:bradley} and \eqref{eqn:cooper} is of order $O(y^{k\ideg_R + (l-k)\ideg_{R'}-1})$, simply because the difference $f(y - a) - f(y-b)$ has degree strictly smaller than $f(y)$ for any rational function $f$ and any constants $a,b$. This proves the bound on the third line of \eqref{eqn:degreeimprove}. When $ l = k$, relations \eqref{eqn:bradley} and \eqref{eqn:cooper} hold up to terms of order:
$$
O \left( y^{k\ideg_R - 2} \right) \qquad \text{instead of} \qquad O \left( y^{k\ideg_R - 1} \right)
$$
because \eqref{eqn:limit} holds up to degree $-2$ instead of just $-1$. This yields the slightly stronger bound on the second line of \eqref{eqn:degreeimprove}. 

\end{proof}

\begin{proof}{\bf of Lemma \ref{lem:magic}:} Recall from \eqref{eqn:defshuf} that shuffle elements are of the form:
\begin{equation}
\label{eqn:mad}
R(z_1,...,z_k) = \frac {\rho(z_1,...,z_k)}{\prod_{1\leq i \neq j \leq k} (z_i-z_j+t)}
\end{equation}
where $\rho$ satisfies the wheel conditions. If we recall the definition of the slope conditions \eqref{eqn:ideg1}, the problem reduces to proving the dimension estimates:
$$
\dim S \ \leq \ \sum_{s\in \BN} \# \Big \{(k_1,e_1),...,(k_s, e_s), \ k_1+...+k_s=k, \ 0 \leq e_i \leq d_{k_i} \Big \}
$$
where:
$$
S = \Big \{ \rho \text{ satisfying \eqref{eqn:wheel} and }i\ideg_{\rho} \leq d_i + 2i(k-i)\ \forall \ i \in \{1,...,k\} \Big \}
$$
The term $2i(k-i)$ that we added to the right hand side of the latter inequality comes from the $i\ideg$ of the denominator of \eqref{eqn:mad}. For any partition $\la = (k_1 \geq ... \geq k_s)$ of $k$, consider the linear maps $\Phi_\lambda : S \longrightarrow \BF[y_1,...,y_s]$ defined by:
$$
\Phi_\lambda(\rho) = \rho \Big( y_1 - t_1, y_1 - 2t_1,...,y_1 - k_1 t_1, \ ... \ , y_s - t_1, y_s - 2 t_1,..., y_s - k_s t_1 \Big)
$$
We will use these linear maps to construct the following ``Gordon filtration" of $S$:
$$
S_\lambda \ = \ \bigcap_{\mu > \lambda} \Phi_\mu^{-1}(0)
$$
where $>$ denotes the dominance ordering on partitions. The required dimension estimates \eqref{eqn:dimest} therefore follow from the statement:
\begin{equation}
\label{eqn:mall}
\dim \Phi_\lambda (S_\la) \leq \ \# \Big \{(k_1,e_1),..., (k_s, e_s) \text{ where } 0 \leq e_i \leq d_{k_i} \Big \}
\end{equation}
where in the right hand side we count unordered collections. To prove \eqref{eqn:mall}, consider any element $\rho \in S_\la$ and define:
$$
p(y_1,...,y_s) := \Phi_\la(\rho) 
$$
The wheel conditions \eqref{eqn:wheel} imply that $p$ vanishes at: \\

\begin{itemize}

\item $y_i - at_1 - t = y_j - a't_1$ \ for all $1 \leq a < k_i$ and $1 \leq a' \leq k_j$ \\

\item $y_i - at_1 + t = y_j - a't_1$ \ for all $1 < a \leq k_i$ and $1 \leq a' \leq k_j$

\end{itemize}

\tab 
for all $i<j$. Meanwhile, the condition that $\Phi_\mu(\rho) = 0$ for all $\mu > \la$ implies that $p$ vanishes at the specializations: \\

\begin{itemize}

\item $y_i - (k_i+1) t_1 = y_j -a' t_1$ \quad for all $1 \leq a' \leq k_j$ \\

\item $y_i = y_j - a' t_1$ \qquad \qquad \qquad \ for all $1 \leq a' \leq k_j$

\end{itemize}

\tab 
for all $i<j$, because each of these specializations entails evaluating $\Phi_\mu(\rho)$ for some $\mu > \la$. The above vanishings are counted with the correct multiplicities, hence:
$$
p(y_1,y_2,...)  = p'(y_1,y_2,...)\prod_{i<j} \left[\prod^{1\leq a \leq k_i}_{1\leq a' \leq k_j} (y_i-y_j-*) \prod^{1\leq a \leq k_i}_{1\leq a' \leq k_j} (y_i-y_j-*') \right] 
$$
for some polynomial $p'$, where $*$ and $*'$ denote various constants arising from the four bullets above. By the condition \eqref{eqn:ideg1}, we conclude that the degree of $p$ in $y_i$ is at most $d_{k_i} + 2k_i(k-k_i)$, for all $i \in \{1,...,s\}$. Therefore, the degree of $p'$ in $y_i$ is at most $d_{k_i}$, which proves the estimate \eqref{eqn:mall}. Note that we must take unordered collections in \eqref{eqn:mall}, since the polynomials $p$ and $p'$ are symmetric in $y_i$ and $y_j$ if $k_i = k_j$. The same symmetry must hold for their ratio $\frac p{p'}$, which thus has one more linear condition on its coefficients and hence one less degree of freedom. \\
\end{proof}

\begin{proof}{\bf of Proposition \ref{prop:degree}:} Write $P_{k,d} = \sym \ A$, where $A$ is the explicit non-symmetric rational function that appears in \eqref{eqn:pkd}. We will prove the stronger statement that $l\ideg_A < \frac {dl}k$. To do so, we must estimate:
\begin{equation}
\label{eqn:mercury}
\deg_y \left(A \Big|_{z_{i_1} \mapsto y - t_1,..., z_{i_l} \mapsto y - l t_1} \right)
\end{equation}
for any subset $S = \{i_1 < ... < i_l\} \subset \{1,...,k\}$. The reason why we only encounter contributions from those subsets where $i_1,...,i_l$ are in increasing order is the fact that $\zeta(-t_1)=0$. Let us divide $S$ into groups of consecutive integers:
\begin{equation}
\label{eqn:chunk}
S = \{x_1+1,...,y_1,x_2+1,...,y_2,...,x_s+1,...,y_s\}
\end{equation}
for some $0 \leq x_1<y_1<...<x_s<y_s \leq k$. Then we see that the degree in \eqref{eqn:mercury} is:
\begin{equation}
\label{eqn:queen}
= \sum_{i=1}^s \left( \left \lfloor \frac {dy_i}k \right \rfloor - \left \lfloor \frac {dx_i}k \right \rfloor \right) - 2s+\delta_{x_1}^0 + \delta_{y_s}^k
\end{equation}
The sum arises from the numerator of \eqref{eqn:pkd}, while the remaining terms come from the denominator. Term by term, it is elementary to prove that expression \eqref{eqn:queen} is:
$$
<\sum_{i=1}^s \left( \frac {dy_i}k - \frac {dx_i}k \right) = \frac {dl}k
$$
\end{proof}

\begin{proof}{\bf of Proposition \ref{prop:basis}:} By the dimension estimate \eqref{eqn:dimest}, it is enough to show that the shuffle elements \eqref{eqn:basis} are linearly independent. We arrange the products \eqref{eqn:basis} in lexicographic order with respect to:
$$
\Gamma = \{(k_1,d_1),...,(k_s,d_s)\} \ > \ \Gamma' = \{(k_1',d_1'),...,(k'_{s'},d'_{s'}) \}
$$
if $d_1/k_1 > d_1'/k_1'$ or if $d_1/k_1 = d_1'/k_1'$ and $k_1<k_1'$. If $(k_1,d_1) = (k_1',d_1')$, then we look at the second elements of $\Gamma$ and $\Gamma'$ to determine their lexicographic ordering, and so on. Therefore, assume for the purpose of contradiction that we have a relation:
\begin{equation}
\label{eqn:fluffy}
P_\Gamma = \sum_{\Gamma'<\Gamma} c_{\Gamma'} P_{\Gamma'}
\end{equation}
where $c_{\Gamma'} \in \BF$ denote various constants. We assume the above relation to be minimal in $|\Gamma| = k_1+...+k_s$. If $d_1/k_1 > d_1'/k_1'$  or if $d_1/k_1 = d_1'/k_1'$  and $k_1<k_1'$, then the left hand side of \eqref{eqn:fluffy} has a $k_1\ideg$ exactly equal to $d_1$. Meanwhile, by \eqref{eqn:degree} and \eqref{eqn:ideg2}, the right hand side has a $k_1\ideg$ strictly less than $d_1$, thus contradicting \eqref{eqn:fluffy}. The only situation in which this argument fails is if $(d_1,k_1) = (d_1',k_1')$, in which case taking the highest degree term in the $k_1\ideg$, we obtain a relation:
$$
P_{\Gamma_0} = \sum_{\Gamma_0' < \Gamma_0} c'_{\Gamma'_0} P_{\Gamma'_0}
$$
where $\Gamma_0 = \Gamma \backslash (k_1,d_1)$. This contradicts the minimality of relation \eqref{eqn:fluffy}. \\
\end{proof}

\begin{proof}{\bf of Proposition \ref{prop:half}:} Just like we showed that $\S_{\ll 0} = \bigoplus_{k=0}^\infty \S_{k|\ll 0}$ is a Lie algebra, one can use \eqref{eqn:degreeimprove} to show that $\bigoplus_{k=0}^\infty \S_{k|\ll 1/k}$ is a Lie algebra. This implies that the commutators in the left hand sides of relations \eqref{eqn:heis}--\eqref{eqn:vir} all lie in $\S_{m+n|\ll 1/(m+n)}$. According to Proposition \ref{prop:bound}, to prove the three required relations, one needs to show that the left and right hand sides of each relation have the same shadow. To this end, note that:
\begin{equation}
\label{eqn:shade}
\sh_{\widetilde{B}_{k}} = t_2^{1-k} \qquad \text{and} \qquad \sh_{\widetilde{L}_{k}} = t_2^{-k} \left(\frac y{t_1} - \frac {k+1}2 \right)
\end{equation}
Relation \eqref{eqn:shade} holds because the only term in the symmetrization which survives the evaluation of \eqref{eqn:shadow} is the identity permutation, since $\zeta(-t_1) = 0$. Then \eqref{eqn:sh} implies that: 
$$
\sh_{[\widetilde{B}_{k}, \widetilde{B}_{l}]} = 0
$$
as well as:
$$
\sh_{[\widetilde{L}_{k}, \widetilde{B}_{l}]} = t_2^{1-k-l} \left(\frac {y}{t_1} -  \frac {k+1}2 \right) - t_2^{1-k-l} \left( \frac {y}{t_1} - l - \frac {k+1}2 \right) = l \cdot \sh_{\widetilde{B}_{k+l}} 
$$
and:
$$
\sh_{[\widetilde{L}_{k}, \widetilde{L}_{l}]} = t_2^{-k-l}  \left( \frac {y}{t_1} -  \frac {k+1}2 \right)\left( \frac {y}{t_1} - k - \frac {l+1}2 \right) -
$$
$$
- t_2^{-k-l}  \left(\frac {y}{t_1} - l - \frac {k+1}2 \right)\left(\frac {y}{t_1} - \frac {l+1}2 \right) =
$$
$$
=  t_2^{-k-l}  \left(\frac {y(l-k)}{t_1} - \frac {l(l-1)}2 + \frac {k(k-1)}2 \right) = (l-k) \cdot \sh_{\widetilde{L}_{k+l}}
$$
thus proving \eqref{eqn:heis}--\eqref{eqn:vir}. \\
\end{proof}

\begin{proof}{\bf of Proposition \ref{prop:h0 h1 h2}:} The proposition follows from the expansion:
$$
\frac {\zeta(w-z)}{\zeta(z-w)} = \frac {(w-z+t_1)(w-z+t_2)(w-z-t)}{(w-z-t_1)(w-z-t_2)(w-z+t)} = 1 -\frac {2t_1t_2t}{w^3} -\frac {6t_1t_2t z }{w^4} + O\left(\frac 1{w^5} \right)
$$
Then taking the coefficients of $\frac 1w$ and $\frac 1{w^2}$ of \eqref{eqn:yang1} gives us:
$$
h_0 * R^+ = R^+ * h_0 \qquad \text{and} \qquad h_1 * R^+ = R^+ * h_1
$$
Meanwhile, taking the coefficient of $\frac 1{w^3}$ of \eqref{eqn:yang1} gives us:
$$
h_2 * R^+ = R^+ * h_2 - (2 k t_1t_2t) \left(\frac {t_1t_2}{-t} \right) \cdot R^+ \quad \Longrightarrow \quad [h_2, R^+ ] = 2 k t_1^2t_2^2 \cdot R^+
$$
Finally, the coefficient of $\frac 1{w^4}$ of \eqref{eqn:yang1} implies:
$$
h_3 * R^+ = R^+ * h_3 - (2k t_1 t_2 t) h_0 * R - (6t_1t_2t) \left( \frac {t_1t_2}{-t} \right) \cdot (z_1+...+z_k) R^+ \Longrightarrow
$$
\begin{equation}
\label{eqn:h3}
\Longrightarrow [h_3,R^+] =  6t_1^2t_2^2 \cdot (z_1+...+z_k) R^+ - 2k t_1 t_2 t \cdot h_0 R^+
\end{equation}
This precisely establishes \eqref{eqn:degree operator} and \eqref{eqn:degree operator 2} when the sign is $+$. The case when the sign is $-$ is analogous. Note that by going further and looking at the coefficient of $\frac 1{w^5}$, one obtains the following identity:
$$
[h_4, R] = \pm 2 t_1^2 t_2^2 \cdot \left(6z^2_1+...+6z^2_k+ k(t_1^2+t_1t_2+t_2^2) \right) R \mp 
$$
\begin{equation}
\label{eqn:degree operator 3}
\mp 6 t_1 t_2 t \cdot h_0 (z_1+...+z_k) R \mp 2k t_1 t_2 t \cdot h_1 R 
\end{equation}
for any $R  \in \S_k^\pm$. This identity will be used in the proof of Theorem \ref{thm:full}. \\
\end{proof} 

\begin{proof}{\bf of Theorem \ref{thm:full}:} We already know the fact that formulas \eqref{eqn:Heis}--\eqref{eqn:Vir} hold when $\sgn \ k = \sgn \ l$, as a consequence of Proposition \ref{prop:half}. When either $k$ or $l$ equals $0$, one needs to prove that:
$$
[B_0, B_k]= 0 \ \ \qquad \ \ [B_0, L_k] = 0
$$
$$
[L_0,B_k] = - k B_k \qquad [L_0,L_k] = - k L_k
$$
which are a consequence of the fact that $h_1$ is central and of \eqref{eqn:degree operator}, respectively. Then we claim that it is sufficient to check the remaining relations only for $L_{-2}, L_{-1},L_1,L_2$ and $B_{-1},B_1$. This is a well-known fact about the Heisenberg-Virasoro algebra: these relations form the base case of an induction to compute expressions such as $[L_k,L_l]$ for $k\geq 3$ and $l < 0$ by writing $(k-2) L_k = [L_{k-1},L_1]$ and then using the Jacobi identity to calculate the commutator with $L_l$.

\tab 
Moreover, the commutation relations between $L_{-1},L_1,B_{-1},B_1$ are simply applications of \eqref{eqn:yang3}. Therefore, we only need to prove the following identities:
\begin{equation}
\label{eqn:daddy}
[B_1,L_{-2}] = B_{-1} \qquad \qquad [L_2,B_{-1}] = B_1
\end{equation}
\begin{equation}
\label{eqn:yankee}
[L_1,L_{-2}] = 3L_{-1} \ \quad \qquad [L_2,L_{-1}] = 3L_1
\end{equation}
\begin{equation}
\label{eqn:pensa}
[L_2,L_{-2}] = \frac {2h_2}{t_1^2t_2^2} + \frac {h_0}2 \left (\frac 1{t_1^2} + \frac 1{t_1t_2} + \frac 1{t_2^2} \right) - \frac {h_0^3 t^2}{2t_1^4t_2^4}
\end{equation}
In \eqref{eqn:daddy} and \eqref{eqn:yankee}, we will only prove the first equality, as the second is obtained by transposition. By \eqref{eqn:ll2} and \eqref{eqn:l2}, we have:
$$
[B_1,L_{-2}] = \left[z_-^0, \widetilde{L}_2+\frac {h_0t}{2t_1^2t_2^2} \widetilde{B}_2 \right] = \left[z_-^0, \frac {[z_+^2 , z_+^0]}{2t_1^2t_2^2} + \frac {h_0t [z_+^1, z_+^0]}{2t_1^3t_2^3} \right]
$$
where we use the notation $z_+^d$ and $z_-^d$ to denote positive and negative shuffle elements, respectively. Formula \eqref{eqn:yang3} allows us to compute the commutators:
$$
[B_1,L_{-2}] = \frac {[h_2, z_+^0]}{2t_1^2t_2^2} + \frac {[z_+^2, h_0]}{2t_1^2t_2^2} + \frac {h_0 t [h_1, z_+^0]}{2t_1^3t_2^3} + \frac {h_0 t [z_+^1, h_0]}{2t_1^3t_2^3} = z_+^0 + 0 + 0 + 0 = B_{-1}
$$
where the first equality uses the Jacobi identity and \eqref{eqn:yang3}, while the second equality uses Proposition \ref{prop:h0 h1 h2}. This proves \eqref{eqn:daddy}. Similarly, \eqref{eqn:ll2} and \eqref{eqn:l2} imply:
$$
[L_1,L_{-2}] = \left[\frac {z_-^1}{t_1t_2}, \widetilde{L}_2+\frac {h_0t}{2t_1^2t_2^2} \widetilde{B}_2 \right] = \left[\frac {z_-^1}{t_1t_2}, \frac {[z_+^2 , z_+^0]}{2t_1^2t_2^2} + \frac {h_0t [z_+^1, z_+^0]}{2t_1^3t_2^3} \right] =
$$
$$
= \frac {[h_3, z_+^0]}{2t_1^3t_2^3} + \frac {[z_+^2, h_1]}{2t_1^3t_2^3} + \frac {h_0 t [h_2, z_+^0]}{2t_1^4t_2^4} + \frac {h_0 t [z_+^1, h_1]}{2t_1^4t_2^4} = \left(\frac {3 z_+^1}{t_1t_2} - \frac {2 t_1 t_2 t \cdot h_0 z_+^0}{2t_1^3t_2^3} \right) + \frac {t \cdot h_0 z_+^0}{t_1^2t_2^2}
$$
which is precisely $3L_{-1}$. This proves \eqref{eqn:yankee}. Note that the last equality used Proposition \ref{prop:h0 h1 h2}. Finally, to prove \eqref{eqn:pensa}, we use \eqref{eqn:ll2} and \eqref{eqn:l2} again:

\begin{equation}
\label{eqn:panda}
\left[L_2, L_{-2} \right] = - \left[ \frac {[z_-^2, z_-^0]}{2t_1^2t_2^2} +\frac {h_0t [z_-^1, z_-^0]}{2t_1^3t_2^3},\frac {[z_+^2, z_+^0]}{2t_1^2t_2^2} +\frac {h_0t [z_+^1, z_+^0]}{2t_1^3t_2^3} \right]
\end{equation}
The reason for the $-$ sign in front of the right hand side is that \eqref{eqn:l2} holds in the positive shuffle algebra $\S^+$. Since the negative shuffle algebra is endowed with the opposite multiplication, formulas \eqref{eqn:ll2} and \eqref{eqn:l2} hold in $\S^-$ only up to a factor of $-1$. We will use the following formula for the commutator of commutators:
$$
\Big[ [a_-,b_-],[a_+,b_+] \Big] = \left[ \left[ \left[a_-,a_+ \right], b_- \right], b_+ \right] - \left[ \left[ \left[b_-,a_+ \right], a_- \right], b_+ \right] -
$$
$$
- \left[ \left[ \left[a_-,b_+ \right], b_- \right], a_+ \right] + \left[ \left[ \left[b_-,b_+ \right], a_- \right], a_+ \right] 
$$
and \eqref{eqn:yang3} to evaluate the right hand side of \eqref{eqn:panda}. For brevity, we will not write down those commutators of the form $[h_0,...]$ and $[h_1,...]$, since we know these are zero by Proposition \ref{prop:h0 h1 h2}. With this in mind, \eqref{eqn:panda} becomes:
$$
\left[L_2, L_{-2} \right] = - \frac {\left[ \left[ h_4, z_-^0 \right], z_+^0 \right] - \left[ \left[ h_2, z_-^2 \right], z_+^0 \right] - \left[ \left[ h_2, z_-^0 \right], z_+^2 \right]}{4t_1^4t_2^4} +
$$
$$
- \frac {h_0t \left(2 \left[ \left[ h_3, z_-^0 \right], z_+^0 \right] - \left[ \left[ h_2, z_-^1 \right], z_+^0 \right] - \left[ \left[ h_2, z_-^0 \right], z_+^1 \right] \right)}{4t_1^5t_2^5} - \frac {h_0^2t^2  \left[ \left[ h_2, z_-^0 \right], z_+^0 \right]}{4t_1^6t_2^6} 
$$
Formulas \eqref{eqn:degree operator}, \eqref{eqn:degree operator 2} and \eqref{eqn:degree operator 3} tell us how to evaluate the above commutators:
$$
\left[L_2, L_{-2} \right] = - \frac {\left[ - 12 t_1^2t_2^2 z_-^2 + 6 t_1t_2 t h_0 z_-^1 + \left(2 t_1t_2t h_1 - 2t_1^2t_2^2(t_1^2+t_1t_2+t_2^2) \right) z_-^0, z_+^0 \right]}{4t_1^4t_2^4} -
$$
$$
- \frac {[z_-^2, z_+^0] + [z_-^0, z_+^2]}{2t_1^2t_2^2} - \frac {h_0t \left[ -6t_1^2t_2^2 z_-^1 + 2 t_1t_2th_0 z_-^0 ,z_+^0 \right] }{2t_1^5t_2^5} - \frac {h_0t[z_-^1,z_+^0] + h_0t[z_-^0,z_+^1] }{2t_1^3t_2^3} +
$$
$$
+ \frac {h_0^2t^2  \left[ z^0_-, z_+^0 \right]}{2t_1^4t_2^4} = \frac {3 h_2}{t_1^2t_2^2} - \frac {3th_0h_1}{2t_1^3t_2^3} - \frac {th_0h_1}{2t_1^2t_2^2} + \frac {h_0(t_1^2+t_1t_2+t_2^2)}{2t_1^2t_2^2} -
$$
$$
- \frac {h_2}{t_1^2t_2^2} + \frac {3th_0h_1}{t_1^3t_2^3} - \frac {t^2h_0^3}{t_1^4t_2^4} - \frac {th_0h_1}{t_1^3t_2^3} + \frac {t^2h_0^3}{2t_1^4t_2^4} = \text{right hand side of \eqref{eqn:pensa}}
$$
\end{proof}

\begin{proof}{\bf of Theorem \ref{thm:act}:} It is straightforward to show that the classes $\of$ generate the localized cohomology rings $H_\bu$ (see, for example, Proposition 2.6 of \cite{mod}). Remark \ref{rem:normal order} requires the normal-ordered integral \eqref{eqn:int} to give well-defined actions of the positive and negative shuffle algebras $\S^\pm \curvearrowright H_\bu$ (pending the consistency checks in the two bullets at the end of Remark \ref{rem:normal order}, which we will prove in Proposition \ref{prop:restriction}). In order to show that these two actions glue to an action of:
$$
\wS = \S^+ \otimes \S^0 \otimes \S^- \curvearrowright H_\bu
$$
we need to show that the action respects \eqref{eqn:yang1}--\eqref{eqn:yang3}. To this end, note that:
$$
R^+(Z) \cdot \left( h(w) \cdot \of \right) = \frac {1}{k!} \cdot \frac {t_1t_2}{-t} :\int: \frac {R^+(Z)}{\zeta(Z-Z)} \frac {\zeta(Z-w)}{\zeta(w-Z)} \cdot
$$
$$
 \overline{f(X-Z) \zeta(Z-X)\frac {\zeta(w-X)}{\zeta(X-w)}} \cdot \tau_\bu(Z+t) \frac {\tau_\bu(w+t)}{\tau_{\bu}(w)} = h(w) \cdot \left( R^+(Z) \frac {\zeta(Z-w)}{\zeta(w-Z)} \cdot \of \right)
$$
for any $R^+ \in \S^+$. Note that this matches \eqref{eqn:yang1}, when the above is interpreted as an equality of power series in $w$. Similarly, one shows that \eqref{eqn:yang2} holds. Finally:
$$
z_-^a \cdot \left(z_+^{a'} \cdot \of \right) = - \int_{|z_+| \gg |z_-| \gg 1} z_-^a z_+^{a'} \cdot \frac {(z_+ - z_- + t_1)(z_+ - z_- + t_2)}{(z_+ - z_-)(z_+ - z_- + t)}
$$
\begin{equation}
\label{eqn:oyvey}
\overline{f(X-z_++z_-) \frac {\zeta(z_+-X)}{\zeta(X-z_-)}} \cdot \frac {\tau_\bu(z_++t)}{\tau_\bu(z_-)} \cdot \frac {dz_+}{2\pi i}  \frac {dz_-}{2\pi i}
\end{equation}
where the last term on the first line is simply the explicit formula for $\zeta(z_+-z_-)$. One shows that $z_+^{a'} \cdot (z_-^{a} \cdot \of )$ is given by the same integrand, but the order of the contours is switched. Therefore, the difference $[z_-^a, z_+^{a'}]\cdot \of$ is given by the residues of \eqref{eqn:oyvey} when $z_+$ passes over $z_-$. Explicitly, the poles one picks up are:
$$
z_+ = z_- \quad \text{with residue} \quad \frac {t_1t_2}{-t} \int_{|z_-| \gg 1} z_-^{a+a'} \overline{f(X) \frac {\zeta(z_- - X)}{\zeta(X - z_-)}} \cdot \frac {\tau_\bu(z_-+t)}{\tau_\bu(z_-)} \frac {dz_-}{2\pi i}
$$
and:
$$
z_+ = z_- - t \quad \text{ with residue } \quad \int_{| z_+ | \gg 1} (z_- - t)^{a} z_-^{a'} \cdot
$$
$$
\overline{f(X - (z_- - t) + z_-) \frac {\zeta(z_- - t - X)}{\zeta(X - z_-)}} \cdot \frac {\tau_\bu(z_-)}{\tau_\bu(z_-)} \cdot \frac {dz_-}{2\pi i}
$$
The first residue yields precisely $h_{a+a'} \cdot \of$, according to \eqref{eqn:cartan}. The second residue is 0, because the ratio of $\zeta$'s and $\tau$'s cancels out (as a consequence of \eqref{eqn:identity}), meaning that there are no poles inside the contour of $z_+$. This matches \eqref{eqn:yang3}. 

\end{proof}

\begin{proof}{\bf of Proposition \ref{prop:restriction}:} Consider first the case when $k=1$ and $R^+  = z^a_+$ for some $a\in \BN$. For any $d\in \BN$, formula \eqref{eqn:int k=1} states that:
\begin{equation}
\label{eqn:modern}
z^a_+ \cdot \overline{f}_d = \res_{\infty} \Big[ z^a \cdot \overline{f(X-z) \zeta(z-X)}_{d+1} \cdot \tau_\bu(z+t) dz \Big]
\end{equation}
Pick a symmetric polynomal $f = f(X)$ such that $f(\bnu) = \delta_\bnu^\bmu$ for all $r$--partitions $\bnu$ of size $d$, which by \eqref{eqn:eqloc} and \eqref{eqn:tautrest} means that $\overline{f}_d = |\bmu \rangle$. As $\bmu$ varies over all $r$--partitions of size $d$, such polynomials give rise to a basis of $H_{\bu,d}$. Then we obtain:
$$
\langle \bla | z^a_+ | \bmu \rangle = \res_{\infty} \Big[ z^a f(\bla-z) \cdot \prod_{\sq \in \bla} \zeta(z-\chi_\sq) \cdot \tau_\bu(z+t) dz \Big] = 
$$
$$
= \res_{\infty} \left[ z^a f(\bla-z) \cdot \frac {\prod^{\sq \text{ inner}}_{\text{corner of }\bla} (z - \chi_\sq + t)}{\prod^{\sq \text{ outer}}_{\text{corner of }\bla} (z - \chi_\sq + t)} dz \right]
$$
where in the last equality we used \eqref{eqn:formula}. Instead of taking the residue of the above expression around $z = \infty$, we may add up the residues around the finite poles, which are of the form $z = \chi_\sq - t$ for an outer corner of $\bla$. Such a value for $z$ is precisely the weight of a removable corner of $\bla$, i.e. $z = \chi_\bsq$ for $\bsq = \blanu$, where $\bla \supset \bnu$ is an $r$--partition. We conclude that:
\begin{equation}
\label{eqn:talking}
\langle \bla | z^a_+ | \bmu \rangle = \sum_{\bla = \bnu + \bsq} \chi_\bsq^a f(\bnu) \cdot \frac {\prod^{\sq \text{ inner}}_{\text{corner of }\bla} (\chi_\bsq-\chi_\sq + t)}{\prod^{\sq \text{ outer}}_{\text{corner of }\bla} (\chi_\bsq - \chi_\sq + t)}
\end{equation}
Note that $f(\bnu) = \delta_\bnu^\bmu$ by our choice of $f$, so $\bnu = \bmu$ is the only summand which appears in the right hand side. However, the above expression is ill-defined, because there is a factor of $0$ in the denominator which we should have removed when we computed the residue. We fix this issue by rewriting the product over corners of $\bla$ into one over corners of $\bmu$, which buys us a factor of $t_1t_2/t$:
$$
\langle \bla | z_+^a | \bmu \rangle = \chi_\bsq^a \cdot  \frac {t_1t_2}t \frac {\prod^{\sq \text{ inner}}_{\text{corner of }\bmu} (\chi_\bsq - \chi_\sq + t)}{\prod^{\sq \text{ outer}}_{\text{corner of }\bmu} (\chi_\bsq - \chi_\sq + t)}
$$
where $\bsq = \blamu$. We observe that the above is precisely equal to  \eqref{eqn:coeff+}. The case when positive shuffle elements are replaced by negative shuffle elements is analogous, so we will leave it as an exercise to the interested reader. 

\tab 
Having shown that \eqref{eqn:int} and \eqref{eqn:coeff+} are equivalent for $k=1$, let us prove the case of general $k$. One can either prove this directly by iterating the residue computation in the previous paragraph $k$ times, or by making the following observation: the normal-ordered integral in \eqref{eqn:int general k} was defined to respect the shuffle product. Since Theorem \ref{thm:generation} ensures that any shuffle element is a linear combination of products of the $z_+^a$, then it is enough to show that formula \eqref{eqn:coeff+} also respects the shuffle product. To this end, note that iterating formula \eqref{eqn:coeff+} for shuffle elements $R,R' \in \S^+$ yields:
$$
\langle \bla | R * R' | \bmu \rangle = \sum_{\bla \supset \bnu \supset \bmu} \langle \bla | R | \bnu \rangle \langle \bnu | R' | \bmu \rangle = \sum_{\bla \supset \bnu \supset \bmu}  R(\blanu) R'(\bnumu)
$$
$$
\cdot \prod_{\bsq \in \blanu} \left[ \frac {t_1t_2}t \prod_{\sq \in \bnu} \zeta(\chi_\bsq - \chi_\sq) \tau_\bu(\bsq+t) \right]  \prod_{\bsq \in \bnumu} \left[ \frac {t_1t_2}t  \prod_{\sq \in \bmu} \zeta(\chi_\bsq - \chi_\sq) \tau_\bu(\bsq+t) \right]  =
$$
$$
\sum_{\bla \supset \bnu \supset \bmu}  R(\blanu) R'(\bnumu) \prod^{\bsq \in \blanu}_{\bsq' \in \bnumu} \zeta(\chi_\bsq - \chi_{\bsq'}) \prod_{\bsq \in \blamu} \left[ \frac {t_1t_2}t  \prod_{\sq \in \bmu} \zeta(\chi_\bsq - \chi_\sq) \tau_\bu(\bsq+t) \right]
$$
The fact that the above matches \eqref{eqn:coeff+} for the shuffle product $R*R'$ is equivalent to the observation that:
$$
(R*R')(\blamu) = \sum_{\bla \supset \bnu \supset \bmu}  R(\blanu) R'(\bnumu) \prod^{\bsq \in \blanu}_{\bsq' \in \bnumu} \zeta(\chi_\bsq - \chi_{\bsq'})
$$
This follows from the fact that $\zeta(-t_1) = \zeta(-t_2) = 0$, which implies that the only terms in the symmetrization \eqref{eqn:shufprod} which survive evaluation at $\{\chi_\sq, \sq \in \blamu\}$ are those such that no variable which enters $R$ corresponds to a box directly below or left of a variable which enters $R'$. This precisely means that the variables of $R$ are specialized to $\{\chi_\sq, \sq \in \blanu\}$ and the variables of $R'$ are specialized to $\{\chi_\sq, \sq \in \bnumu\}$, for some intermediate $r$--partition $\bla \supset \bnu \supset \bmu$. 

\tab 
The case of negative shuffle elements is analogous, so we leave it as an exercise. Having shown that \eqref{eqn:int} is equivalent to \eqref{eqn:coeff+}--\eqref{eqn:coeff-} for any shuffle element $R^\pm$, we conclude that the former formula is well-defined. Indeed, the two bullets in Remark \ref{rem:normal order} require us to show that: \\

\begin{itemize}

\item  the action only depends on $R^\pm$ itself, and not on its presentation as a linear combination of elements \eqref{eqn:shuffle element}, which is obvious from \eqref{eqn:coeff+} \\

\item the action only depends on the cohomology class $\of_d \in H_{\bu,d}$, i.e. on the evaluations $f(\bmu)$ for all $\bmu \vdash d$, and not on the choice of the symmetric function $f$. It is enough to check this for a shuffle element of the form $z_+^a$, in which case it follows because the coefficient of a general $|\bla \rangle$ in the right hand side of \eqref{eqn:modern} is given by the right hand side of \eqref{eqn:talking}.
\end{itemize}
\end{proof}

\begin{proof}{\bf of Proposition \ref{prop:virtualtangent}:} As in \eqref{eqn:puppets}, we need to compute the contribution of the affine space of linear maps $X,Y,A,B$ that enter Definition \ref{def:fine}. This computation is based on the following statement in linear algebra: \\

\begin{claim}
\label{claim}

Consider two vector spaces $E$ and $F$, equipped with partial flags:
$$
0 = E_k \subset ... \subset E_1 \subset E_0 = E
$$
$$
0 = F_k \subset ... \subset F_1 \subset F_0 = F
$$
Then the vector space of linear maps $E \rightarrow F$ which preserve the given flags equals:
\begin{equation}
\label{eqn:claim}
\sum_{i=0}^{k-1} \frac {F_i}{E_i - E_{i+1}} 
\end{equation}
in the Grothendieck group of vector spaces. Recall that we write $\frac FE$ for $E^\vee \otimes F$. 

\end{claim}

\tab 
The proof closely follows that of Claim 6.1 in \cite{tor}, so we leave it as an exercise to the interested reader. According to this claim, the vector space of linear maps $X$ as in  Definition \ref{def:fine} has $K$--theory class given by:
$$
\frac 1{e^{t_1}} \left( \frac {\V_+}{\V_-} + \sum_{i=1}^{k} \frac {\L_1+...+\L_{i-1}}{\L_i} \right) = \frac 1{e^{t_1}} \left( \frac {\V_+}{\V_-} + \sum_{1 \leq j < i \leq k} \frac {\L_j}{\L_i} \right)
$$
The reason for this is the condition that $X$ is nilpotent, which amounts to the fact that $X:V_i \rightarrow V_{i-1}$ in \eqref{eqn:flag}. Similarly, the vector space of linear maps $Y$ has class:
$$
\frac 1{e^{t_2}} \left( \frac {\V_+}{\V_-} + \sum_{i=1}^{k} \frac {\L_1+...+\L_{i-1}}{\L_i} \right) = \frac 1{e^{t_2}} \left( \frac {\V_+}{\V_-} + \sum_{1 \leq j < i \leq k} \frac {\L_j}{\L_i} \right)
$$
while the vector spaces of linear maps $A$ and $B$ have classes:
$$
\sum_{i=1}^r \frac {\V_+}{e^{u_i}} \qquad \qquad \text{and} \qquad \qquad \sum_{i=1}^r \frac {e^{u_i-t}}{\V_-}
$$
respectively. If $X,Y$ both are nilpotent on the flag \eqref{eqn:flag}, i.e. $X,Y:V_i \rightarrow V_{i-1}$, then their commutator satisfies the stronger property that $[X,Y] : V_i \rightarrow V_{i-2}$. Therefore, the differential $d\mu$ of the moment map takes values in the affine space with $K$--theory class:
$$
\frac 1{e^t} \left( \frac {\V_+}{\V_-} + \sum_{i=1}^{k} \frac {\L_1+...+\L_{i-2}}{\L_i} \right) = \frac 1{e^t} \left(\frac {\V_+}{\V_-} + \sum_{1 \leq j < i \leq k} \frac {\L_j}{\L_i} - \sum_{i=1}^{k-1} \frac {\L_{i}}{\L_{i+1}} \right) 
$$
Finally, the Lie algebra of $B_{d_+,d_-}$ simply consists of those endomorphisms which preserve the flag \eqref{eqn:flag}, so it contributes the following $K$--theory class:
$$
\frac {\V_+}{\V_-} + \sum_{i=1}^{k} \frac {\L_1+...+\L_{i}}{\L_i} = \frac {\V_+}{\V_-} + \sum_{1 \leq j \leq i \leq k} \frac {\L_j}{\L_i}
$$
Adding and subtracting the above quantities according to \eqref{eqn:master} gives us \eqref{eqn:tangentclass1}. \\
\end{proof}

\begin{proof}{\bf of Proposition \ref{prop:compare}:} We will treat only the case of $x_m^+$, and leave $x_m^-$ as an exercise to the interested reader. By the localization formula, we have:
\begin{equation}
\label{eqn:virtualreality}
\langle \bla | x_m^+ | \bmu \rangle = \sum^{\text{SYT} \ \Y}_{\text{of shape }\blamu} - m \left(l_1|_\Y,...,l_k|_\Y \right) \cdot \frac {e(T_{\bla} \M_{d_+})}{e(T^\vir_{\Y} \  \fZ_{d_+,d_-})}
\end{equation}
By formulas \eqref{eqn:tangentclass} and \eqref{eqn:tangentclass1}, we obtain the following equality of $K$--theory classes:
$$
\left[T \M_{r,d_+} \right] - \left[ T^\vir \ \fZ_{d_+,d_-} \right]  = \frac k{e^{t_1}} + \frac k{e^{t_2}} - \frac k{e^t} - \sum_{i=1}^{k-1} \frac {\L_{i}}{e^{t}\L_{i+1}} -
$$
$$
- \left(1 - \frac 1{e^{t_1}} \right)\left(1 - \frac 1{e^{t_2}} \right)\left(\sum_{i=1}^k \frac {\V_-}{\L_i} + \sum_{1\leq i < j \leq k} \frac {\L_j}{\L_i} \right) +  \sum_{i=1}^r \frac {e^{u_i-t}}{\L_i}
$$
If we evaluate the above expression at a fixed point of $\fZ_{d_+,d_-}$, which we recall from \eqref{eqn:fixedfine} can be described as a SYT $\Y$ as in \eqref{eqn:label}, we conclude that:
\begin{equation}
\label{eqn:alig}
\left[T_\bla \M_{r,d_+} \right] - \left[ T^\vir_\Y \ \fZ_{d_+,d_-} \right]  = \frac k{e^{t_1}} + \frac k{e^{t_2}} - \frac k{e^{t}} - \sum_{i=1}^{k-1} \frac {e^{\chi_{i}}}{e^{\chi_{i+1}+t}} -
\end{equation}
$$
- \left(1 - \frac 1{e^{t_1}} \right)\left(1 - \frac 1{e^{t_2}} \right)\left(\sum_{i=1}^k \sum_{\sq \in \bmu} \frac {e^{\chi_\sq}}{e^{\chi_i}} + \sum_{1\leq i < j \leq k} \frac {e^{\chi_j}}{e^{\chi_i}} \right) +  \sum_{i=1}^r \frac {e^{u_i-t}}{e^{\chi_i}}
$$
where $\chi_1,...,\chi_k$ denote the weights of the boxes $1,...,k$ of the SYT. Recall that these are none other than the restrictions of the cohomology classes $l_i =c_1(\L_i)$ to the fixed point $\Y$. From the above expression, we can obtain a formula for the ratio of Euler classes \eqref{eqn:virtualreality} by transforming all expressions according to the prescription:
$$
... + e^{x_1} + e^{x_2} + ... - e^{y_1} - e^{y_2} - ... \qquad \mapsto \qquad \frac {... x_1 x_2 ...}{... y_1 y_2 ...}
$$
By this rule, \eqref{eqn:virtualreality} becomes:
$$
\langle \bla | x_m^+ | \bmu \rangle = \sum^{\text{SYT} \ \Y}_{\text{of shape }\blamu} - m(\chi_1,...,\chi_k) \cdot \frac {(-t_1)^k (-t_2)^k }{(-t)^k \prod_{i=1}^{k-1} (\chi_{i} - \chi_{i+1} - t)} 
$$
\begin{equation}
\label{eqn:indahouse}
\prod_{1\leq i < j \leq k} \zeta(\chi_i-\chi_j) \prod_{i=1}^k \left[ \prod_{\sq \in \bmu} \zeta(\chi_i - \chi_\sq) \tau_\bu(\chi_i+t) \right]
\end{equation}
To be precise, the first and second lines of \eqref{eqn:alig} correspond to the factors on the first and second lines, respectively, of \eqref{eqn:indahouse}. Comparing formula \eqref{eqn:indahouse} with Corollary \ref{cor:restriction} implies that the geometric operator $x_m^+$ and the shuffle element $C_m^+$ have the same restrictions in the basis of fixed points, hence they are equal. \\ 
\end{proof}

\begin{proof}{\bf of Proposition \ref{prop:imp}:} We will prove formula \eqref{eqn:imp} in the case when the sign is $\pm = +$, and leave the case of $-$ as an exercise to the interested reader. We may compute the coefficients of the left hand side of \eqref{eqn:imp} in the fixed point basis:
$$
\langle \bla | \pi^{+}_* \Big( \left[ \fZ^\vir_{d_+,d_-} \right] m(l_1,..,l_k) \Big) = \sum^{\text{SYT }\Y}_{\text{of upper shape }\bla} m \left(l_1|_\Y,...,l_k|_\Y \right) \cdot \frac {e(T_{\bla} \M_{d_+})}{e(T^\vir_{\Y} \  \fZ_{d_+,d_-})}
$$
where we call $\bla$ the upper shape of an almost standard Young tableau of shape $\blamu$. By analogy with \eqref{eqn:indahouse}, we see that the right hand side of the above expression is:
$$
\langle \bla | \pi^{+}_* \Big( \left[ \fZ^\vir_{d_+,d_-} \right] m(l_1,..,l_k) \Big)  = - \left(\frac {t_1t_2}t \right)^k \sum^{\text{SYT }}_{\text{of upper shape }\bla}  m(\chi_1,...,\chi_k) 
$$
\begin{equation}
\label{eqn:want}
\frac {\prod_{i < j} \zeta(\chi_i-\chi_j)}{\prod_{i=1}^{k-1} (\chi_{i+1} - \chi_{i}+t)}  \prod_{i=1}^k \left[ \prod_{\sq \in \bmu} \zeta(\chi_i - \chi_\sq) \tau_\bu(\chi_i+t) \right]
\end{equation}
where $\chi_1,...,\chi_k$ denote the weights of the boxes labeled $1,...,k$ in a standard Young tableau. We will now compute the restriction of integral in the right hand side of \eqref{eqn:imp} to the fixed point $\langle \bla |$ and show that it equals \eqref{eqn:want}. We have:
\begin{equation}
\label{eqn:desire}
\langle \bla | \text{RHS of \eqref{eqn:imp}} = - \int_{|z_k| \gg ... \gg |z_1| \gg X} \frac {dz_1}{2\pi i} \ ... \ \frac {dz_k}{2\pi i}
\end{equation}
$$
\frac {m(z_1,...,z_k)\prod_{i<j} \zeta(z_j-z_i)^{-1}}{\prod_{i=1}^{k-1} (z_{i+1} - z_{i} + t)} \prod_{i=1}^k \frac {\prod^{\bsq \text{ inner}}_{\text{corner of }\bla} (z_i - \chi_\bsq + t)}{\prod^{\bsq \text{ outer}}_{\text{corner of }\bla} (z_i - \chi_\bsq + t)} $$
where we used \eqref{eqn:formula}. We will first integrate the above in the variable $z_1$, and note that the residues are of the form $z_1 = \chi_1 := \chi_{\sq_1}$ for a removable box $\sq_1 \in \bla$ (a removable box is one unit southwest of an outer corner). Let us write:
$$
\bnu_1 = \bla \backslash \sq_1
$$
which allows us to rewrite \eqref{eqn:desire} as:
$$
\langle \bla | \text{RHS of \eqref{eqn:imp}} = - \sum_{\bnu_1 \subset \bnu_0 = \bla} \int_{|z_k| \gg ... \gg |z_2| \gg X} \frac {dz_2}{2\pi i} \ ... \ \frac {dz_k}{2\pi i}
$$
$$
m(\chi_1,z_2,...,z_k) \cdot \frac {\prod_{i=2}^k \zeta (z_i - \chi_1)^{-1} \prod_{2 \leq  i < j \leq k} \zeta(z_j-z_i)^{-1}}{(z_2 - \chi_1+t) \prod_{i=2}^{k-1} (z_{i+1}-z_{i}+t)} 
$$
$$
\frac {\prod^{\bsq \text{ inner}}_{\text{corner of }\bla} (\chi_1 - \chi_\bsq + t)}{\prod^{\bsq \text{ outer}}_{\text{corner of }\bla} (\chi_1 - \chi_\bsq + t)} \prod_{i=2}^k \frac {\prod^{\bsq \text{ inner}}_{\text{corner of }\bla} (z_i - \chi_\bsq + t)}{\prod^{\bsq \text{ outer}}_{\text{corner of }\bla} (z_i - \chi_\bsq + t)} 
$$
The key observation now is that the first product on the second line can change the second product on the third line to one over the corners of $\bnu_1$ instead of $\bla$, once again by using \eqref{eqn:formula}:
$$
\langle \bla | \text{RHS of \eqref{eqn:imp}} = - \sum_{\bnu_1 \subset \bnu_0 = \bla} \int_{|z_k| \gg ... \gg |z_2| \gg X} \frac {dz_2}{2\pi i} \ ... \ \frac {dz_k}{2\pi i}
$$
$$
\frac {m(\chi_1,z_2,...,z_k) \prod_{2 \leq  i < j \leq k} \zeta(z_j-z_i)^{-1}}{(z_2 - \chi_1+t) \prod_{i=2}^{k-1} (z_{i+1}-z_{i}+t)}
$$
$$
\frac {\prod^{\bsq \text{ inner}}_{\text{corner of }\bla} (\chi_1 - \chi_\bsq + t)}{\prod^{\bsq \text{ outer}}_{\text{corner of }\bla} (\chi_1 - \chi_\bsq + t)} \prod_{i=2}^k \frac {\prod^{\bsq \text{ inner}}_{\text{corner of }\bnu_1} (z_i - \chi_\bsq + t)}{\prod^{\bsq \text{ outer}}_{\text{corner of }\bnu_1} (z_i - \chi_\bsq + t)} 
$$
One may then ask for the finite residues of the above expression in $z_2$. The exact same argument as before indicates that these arise at the poles $z_2 = \chi_2 := \chi_{\sq_2}$ where $\sq_2$ is a removable box of $\bnu_1$. We may write $\bnu_2 = \bnu_1 \backslash \sq_2$. Repeating this argument for the variables $z_3,...,z_k$ gives rise to a flag of partitions:
$$
\bmu = \bnu_k \subset \bnu_{k-1} \subset ... \subset \bnu_1 \subset \bnu_0 = \bla
$$
which is the same information as a standard Young tableau. Therefore, the above residue computation yields:
$$
\langle \bla | \text{RHS of \eqref{eqn:imp}} =  - \sum^{\text{SYT}}_{\bmu = \bnu_k \subset \bnu_{k-1} \subset ... \subset \bnu_1 \subset \bnu_0 = \bla} 
$$
$$
\frac {m(\chi_1,...,\chi_k)}{\prod_{i=1}^{k-1} (\chi_{i+1} - \chi_i + t)} \prod_{i=1}^k \frac {\prod^{\bsq \text{ inner}}_{\text{corner of }\bnu_{i-1}} (\chi_i - \chi_\bsq + t)}{\prod^{\bsq \text{ outer}}_{\text{corner of }\bnu_{i-1}} (\chi_i - \chi_\bsq + t)} 
$$
One may use formula \eqref{eqn:formula} to change the product over corners of $\bnu_{i-1}$ into a product over corners of the lower shape $\bmu$ of the resulting SYT. This yields precisely formula \eqref{eqn:want}, and concludes the proof of the Proposition. \\
\end{proof}

\end{document}